\documentclass[twocolumn]{autart}    
\usepackage{amsmath,bm}
\usepackage{color}
\usepackage{multirow}
\usepackage{mathrsfs}
\usepackage{tcolorbox}
\usepackage{subfigure}
\usepackage{dsfont}
\usepackage{amssymb}
\usepackage{setspace}
\usepackage{graphicx}
\usepackage{float}
\newtheorem{remark}{Remark}
\newtheorem{lemma}{Lemma}

\newtheorem{assumption}{Assumption}

\newtheorem{proposition}{Proposition}

\newtheorem{proof}{\textbf{Proof}}

\def\begcen{\begin{center}}
\def\endcen{\end{center}}

\newcommand{\col}{\mbox{col}}

\def\bw{{\bf w}}
\def\bx{{\bf x}}
\def\bv{{\bf v}}
\def\bi{{\bf i}}
\def\bc{{\bf c}}

\def\call{{\mathcal L}}

\def\hal{\frac{1}{2}}

\def\L2e{{\mathcal L}_{2e}}

\def\rea{\mathbb{R}}

\def\et{\epsilon_t}

\def\l2{{\mathcal L}_2}
\def\l2e{{\cal L}_{2e}}

\def\rea{\mathbb{R}}

\def\begequarr{\begin{eqnarray}}
\def\endequarr{\end{eqnarray}}
\def\begequarrs{\begin{eqnarray*}}
\def\endequarrs{\end{eqnarray*}}
\def\begarr{\begin{array}}
\def\endarr{\end{array}}
\def\begequ{\begin{equation}}
\def\endequ{\end{equation}}
\def\lab{\label}
\def\begdes{\begin{description}}
\def\enddes{\end{description}}
\def\begenu{\begin{enumerate}}
\def\begite{\begin{itemize}}
\def\endite{\end{itemize}}
\def\endenu{\end{enumerate}}

\def\lef[{\left[\begin{array}}
\def\rig]{\end{array}\right]}
\def\qed{\hfill$\Box \Box \Box$}
\def\begcen{\begin{center}}
\def\endcen{\end{center}}
\def\begrem{\begin{remark}\rm}
\def\endrem{\end{remark}}

\def\TAC{{\it IEEE Trans. Automatic Control}}
\def\TIE{{\it IEEE Trans. Industrial Electronics}}

\def\AUT{{\it Automatica}}

\def\CEP{{\it Control Engineering Practice}}

\def\begmat#1{\begin{bmatrix}#1\end{bmatrix}}
\def\begali#1{\begin{align}{#1}\end{align}}
\def\begalis#1{\begin{align*}{#1}\end{align*}}

\usepackage{color}

\begin{document}

\begin{frontmatter}

\title{A Globally Exponentially Stable Position Observer for Interior Permanent Magnet Synchronous Motors} 

\thanks[footnoteinfo]{Corresponding author: Bowen Yi. }
\author[L2S,ITMO]{Romeo Ortega}\ead{ortega@lss.supelec.fr},
\author[SJTU,L2S]{Bowen Yi}\ead{b.yi@outlook.com},               
\author[belgrade]{Slobodan Vukosavic}\ead{boban@etf.rs},
\author[POSTECH]{Kwanghee Nam}\ead{kwnam@postech.ac.kr},
\author[POSTECH]{Jongwon Choi}\ead{jongwon@postech.ac.kr}

\address[L2S]{Laboratoire des Signaux et Syst\`emes, CNRS-CentraleSup\'elec, Gif-sur-Yvette 91192, France}
\address[ITMO]{Department of Control Systems and Informatics, ITMO University, Saint Petersburg 197101, Russia}
\address[SJTU]{Australian Centre for Field Robotics, The University of Sydney, Sydney, NSW 2006, Australia}
\address[belgrade]{Electrical Engineering Department, University of Belgrade, Belgrade 11000, Serbia}
\address[POSTECH]{Department of Electrical Engineering, Pohang University of Science and Technology
(POSTECH), Pohang 790-784, Korea}


\begin{abstract}                          
The design of a position observer for the interior permanent magnet synchronous motor is a challenging problem that, in spite of many research efforts, remained open for a long time. In this paper we present the first {\em globally exponentially convergent} solution to it, assuming that the saliency is not too large. As expected in all observer tasks, a persistency of excitation condition is imposed. Conditions on the operation of the motor, under which it is verified, are given. In particular, it is shown that at {\em rotor standstill}---when the system is not observable---it is possible to inject a probing signal to enforce the persistent excitation condition. {The  high performance of the proposed observer, in standstill and high speed regions, is verified by extensive series of test-runs on an experimental setup.}
\end{abstract}
\end{frontmatter}

%
\section{Introduction}
\label{sec1}
%
Electrical motors are benchmark systems that have been intensively studied by control researchers, who have produced more than six research monographs on the topic \cite{CHI,GLUDEL,KHOetalbook,MARetalbook,NAMbook,ORTetalbook} and hundreds of publications in our leading research journals in systems and control.  One of the most challenging problems that appears in this field is the so-called {\em sensorless control}, that is, the control of the motor measuring only the electrical coordinates. The problem has an enormous practical and economical significance as the shaft-sensorless operation voids the need to install, run and interface a dedicated shaft sensor and, theoretically, it requires the design of an observer for a highly nonlinear system. Many well-known control theorists have contributed to the mathematical solution of this problem, mainly for induction \cite{GLUDEL,MARetal} and {\em surface-mount} [also called non-salient] permanent magnet synchronous motors (SPMSMs) \cite{BERPRAaut,CHOetaltpe,ORTetalcst,TOMVER,VERetal}. For a review of the literature, the reader is referred to the aforementioned research monographs and papers.

Because of the reluctance torque and higher power density, as well as their cheaper production cost, it has been recognized in recent years that {\em interior} [also called salient] PMSMs, are more suitable for industrial and home appliances than SPMSMs or induction motors---becoming the {\em de facto} standard in these applications \cite[Subsection 6.1.4]{NAMbook}. The dynamic equations that describe the behavior of IPMSMs are far more complicated than those of SPMSMs. Indeed, because of the rotor saliency, the model {must} incorporate the effect of self and mutual inductances, which vary with an electrical angle between phases and rotor axis---see the discussions in \cite[Subchapter 6.2]{NAMbook} and \cite[Section VI]{ORTetalcst}. In spite of an intensive research activity in the industrial electronics and the control theory communities the problem of designing a globally stable observer for IPMSMs has remained open for many years. This fact is openly recognized by one of the leading authorities in observer design in \cite[Section 6]{BERPRA} where it is stated:

{\em  ``Extension to {salient} models. We are unaware of any observer for this case".}

 The main purpose of this paper is to present the first globally convergent observer for IPMSMs.

The observer proposed in this paper is a gradient-descent search---an approach first proposed in observer theory in \cite{SHI} and applied for the first time to PMSMs in \cite{ORTetalcst}. Unfortunately, the simple construction proposed in \cite{ORTetalcst} is not applicable for IPMSMs. Indeed, in  \cite{ORTetalcst} it is shown that flux and current in SPMSMs verify an algebraic relation, which is {\em independent} of the mechanical coordinates. From this algebraic relation a quadratic criterion to be minimized---whose gradient is computable from the electrical coordinates---can be easily constructed. Unfortunately, this algebraic relation in IPMSMs depends on the {\em rotor position}---a fact that is discussed in Section \ref{sec20}.

To be able to construct a gradient descent-based observer for IPMSMs it was proposed in \cite{CHOetaltie} to follow the approach pursued in \cite{BOBetal} for SPMSMs, namely, to derive {\em via filtering}, a linear regression equation for the flux. In  \cite{CHOetaltie} it was shown that, when applied to the \emph{active flux} of the IPMSM \cite{BOLetal}, the procedure of  \cite{BOBetal} yields an {\em additively perturbed} linear regression. Proceeding from this regression, and neglecting the disturbance term, a gradient-descent search observer was proposed in  \cite{CHOetaltie}. Although experimental evidence proved the good performance of this observer, its theoretical analysis was hampered by the presence of the neglected perturbation term. In \cite{CHONAM} another observer that takes into account the presence of the disturbance was proposed. Extensive experimental evidence proved the high performance of this observer---incorporating at low speeds the signal injection feature commonly used in sensorless control \cite{NAMbook}. However, because of the complexity of the observer dynamics, it was not possible to carry-out the stability analysis. See Section \ref{sec6} for a discussion on this matter.

In this paper we propose a modification to the observer proposed in \cite{CHONAM} for which a complete theoretical analysis allows us to establish its {\em global exponential stability} (GES). As usual in all observer tasks, a persistency of excitation (PE) condition is imposed. Conditions on the operation of the motor, under which it is verified, are given. It is shown that at {\em standstill}---when the flux is not observable \cite{KOTetal}---PE is enforced injecting a probing signal as done in \cite{CHONAM,NAMbook,YIetal}.

The remainder of the paper is organized as follows.  In Section \ref{sec20} we present the model of the PMSMs and explain why the approach adopted in \cite{ORTetalcst} is not applicable to IPMSMs. Section \ref{sec2} presents the linear regression representation from \cite{CHONAM}, that is the basis for our observer design, which is given in Section \ref{sec3}. In Section \ref{sec4} we discuss the PE assumption required by our main result.  In Section \ref{sec5} we present some simulation and experimental results, which illustrate the {performance}---and limitations---of the proposed observer, as well as its operation at standstill injecting a probing signal. The paper is wrapped-up in  Section \ref{sec6} with some concluding remarks and future research, including a discussion on the observer given in \cite{CHONAM}.

\begin{tcolorbox}
\begin{center}
 {\bf Nomenclature}
\end{center}
\vspace{0.2cm}
  \renewcommand\arraystretch{1.4}
\small
\begin{tabular}{ll}
$\alpha\beta$ & Stationary axis reference frame quantities\\
$\bv,\bi \in \rea^2$ & Stator voltage and current [V, A] \\
$\lambda \in \rea^2$ & Stator flux [Wb]\\
$\bx \in \rea^2$ & Active flux [Wb] \\
$\theta \in {\mathbb S}$ & Rotor flux angle [rad] \\
$R$    & Stator winding resistance [$\Omega$]\\
$\psi_m$ & PM flux linkage constant  \\
$L_d,L_q$ & $d$ and $q$-axis inductances [H] \\
$L_0$ & Inductance difference $L_0:= L_d - L_q$ [H] \\
$L_s$ & {Averaged inductance} $L_s:= {L_d + L_q \over 2}$ [H] \\
$|\cdot|$ & Euclidean norm of a vector \\
$p$ & Differential operator $p:= {d \over dt}$\\
$G(p)[w]$ & Action of $G(p) \in \rea(p)$ on a signal $w(t)$\\
$\nabla_x$ & Gradient transpose $\nabla_x :=\big({\partial \over \partial x}\big)^\top$\\
$\tilde{(\cdot)}$ & Estimation error defined as $\hat{(\cdot)}-(\cdot)$ \\
$I_{n}$ & $n\times n$ identity matrix 
\end{tabular}
\end{tcolorbox}
%
\section{Difference Between SPMSM and IPMSM}
\lab{sec20}
%
In this section we present the mathematical models of the SPMSM and IPMSM and explain why the simple approach, proposed in \cite{ORTetalcst}  for observer design of the former, is not applicable for the IPMSM. For both motors the magnetic energy stored within magnetic circuits is given as
$$
H_E(\lambda,\theta) = \hal [\lambda - \psi_m \bc(\theta)]^\top \call^{-1}(\theta) [\lambda - \psi_m \bc(\theta)],
$$
where $\call(\theta) \in \rea^{2 \times 2}$ is the generalized inductance matrix, defined as
\begin{equation*}
\label{sys}
\call(\theta)=\left\{
\begin{aligned}
 & L_s I_2 & \mbox{for\;the\;SPMSM}\;(L_d=L_q)\\
 & \big[L_s I_2 + {L_0 \over 2} Q(2\theta) \big]& \mbox{for\;the\;IPMSM}\;(L_d \neq L_q),
\end{aligned}
\right.
\end{equation*}
where
$$
Q(2 \theta) := \lef[{cc} \cos(2 \theta) & \sin(2\theta) \\   \sin(2\theta) &  -\cos(2
\theta)\rig],
$$
and we defined $\bc(\theta):= \col(\cos\theta, \sin\theta)$. The electrical dynamics (in the stationary $\alpha\beta$ frame) is given by Faraday's Law
\begequ
\label{ipmsm1}
\begin{aligned}
  \dot{\lambda} & = - R\bi+ \bv
\end{aligned}
\endequ
with the constitutive relation
$$
\bi=\nabla_\lambda H_E(\lambda,\theta)=\left\{
\begin{aligned}
 & {1 \over L_s}[\lambda - \psi_m \bc(\theta)]  & \mbox{(SPMSM)} \\
 & \call^{-1}(\theta) [\lambda - \psi_m \bc(\theta)]  &  \mbox{(IPMSM)}.
\end{aligned}
\right.
$$
We underline the fact that $L_0=0$ for SPMSM considerably simplifying the equations.

Noting that the SPMSM verifies the algebraic constraint
\begequ
\lab{algcon}
|\lambda - L_s \bi|^2 - \psi_m^2=0,
\endequ
the flux observer for SPMSM proposed in \cite{ORTetalcst} is a {\em gradient-descent} search for the minimization of the quadratic criterion
$$
J(\hat \lambda):={1 \over 4}\Big(|\hat \lambda - L_s \bi|^2 - \psi_m^2\Big)^2,
$$
leading to
\begali{
\dot {\hat \lambda}
& = \bv - R \bi -\gamma (|\hat \lambda - L_s \bi|^2 - \psi_m^2)(\hat \lambda - L_s \bi).
\lab{obsort}
}
As shown in \cite{ORTetalcst} the flux observer \eqref{obsort} has some remarkable stability properties, and its excellent performance has been validated experimentally \cite{NAMbook}. See also \cite{MALetal} where it is shown that the following slight variation of \eqref{obsort}
$$
\dot {\hat \lambda}= \bv - R \bi -\gamma\max\{0,|\hat \lambda - L_s \bi|^2 - \psi_m^2\}(\hat \lambda - L_s \bi)
$$
ensures global convergence.

Unfortunately, the approach proposed above is not applicable to IPMSMs. Indeed, although the IPMSM still verifies an algebraic constraint similar to \eqref{algcon}, that is
$$
|\lambda - \call(\theta) \bi|^2 - \psi_m^2=0,
$$
it is not possible to compute its gradient without the knowledge of $\theta$. For this reason, in this paper we proceed as done in \cite{BOBetal,CHOetaltie} and look for the generation, via filtering, of a linear regression of $\lambda$.

\section{A Linear Regression Equation of the IPMSM}
\lab{sec2}
%
As indicated in the previous section, the electrical dynamics of the IPMSM is given by Faraday's Law \eqref{ipmsm1}, together with the constitutive relation
\begequ
\label{ipmsm2}
\begin{aligned}
\lambda & = \big[L_s I_2 + {L_0 \over 2} Q(2\theta) \big]\bi +  \psi_m \bc(\theta).
\end{aligned}
\endequ
Some simple calculations {\cite{BERPRA,CHOetaltie}} show that \eqref{ipmsm2} may be written as
\begequ
\label{lamithe}
\begin{aligned}
  \lambda & = L_q \bi + (L_0 \bi^\top \bc(\theta) + \psi_m) \bc(\theta),
\end{aligned}
\endequ
In \cite{CHOetaltie} it is proposed to obtain the rotor angle via the estimation of the \emph{active flux} of the IPMSM, defined in \cite{BOLetal} as
\begequ
\label{x}
\bx:= \lambda - L_q \bi.
\endequ
The motivation to consider this signal is twofold. First, from \eqref{lamithe} and \eqref{x}, we have that
$$
\bx = [L_0\bi^\top \bc(\theta) + \psi_m] \bc(\theta).
$$
Consequently,
\begequ
\lab{xsqu}
|\bx|^2= [L_0\bi^\top \bc(\theta) + \psi_m]^2.
\endequ
We make at this point two assumptions, which are often satisfied in practice and are, therefore, implicitly made in the applications publications \cite{CHOetaltpe,CHOetaltie,CHONAM}. First, since the inductance difference $L_0$, a measure of anisotropy of the machine, is usually {\em very small}, we make the following.

\begin{assumption}\rm
\lab{ass1}
The current $\bi$ verifies $|L_0 \bi | < \psi_m$.
\end{assumption}

Second, to make mathematically rigorous the subsequent analysis, we also need the following.

\begin{assumption}
\rm\label{ass2}
The motor \eqref{ipmsm1}, \eqref{ipmsm2} operates in a mode guaranteeing that all signals $\bi$, $\bv$ and $\lambda$ are bounded and that $|\bx| \ge x_{\min} >0$ .
\end{assumption}

Given these assumptions we can write
$$
{\bx \over |\bx|}= \bc(\theta),
$$
and the rotor angle is easily reconstructed from $\bx$ via
$$
\theta =  {\rm atan2}(\bx_2,\bx_1),
$$
where ${\rm atan2}(\cdot,\cdot)$ is the well-known ``2-argument $\arctan$" function. A second, and most important motivation, is contained in the following lemma, whose proof was established in \cite{CHONAM} and, to make the paper self-contained, it is also given below.

\begin{lemma}
\label{lem1}\rm
The electrical dynamics of the IPMSM \eqref{ipmsm1}, \eqref{ipmsm2} satisfies the following (perturbed) linear regression equation
\begequ
\label{flux_dyn}
\begin{aligned}
y & = \Phi^\top \bx + d + \et
\end{aligned}
\endequ
with the (unknown) perturbing signal $d$ given by
$$
 d := - \ell {\alpha p \over p+ \alpha} \Big[\bi^\top {\bx \over |\bx|}\Big],
$$
where $\alpha>0$ is a {\em tuning} parameter, $\et$ is an exponentially decaying term caused by the initial condition of the filters\footnote{Following standard practice, we neglect these terms in the sequel.}, $\ell:=\psi_m L_0$, and we defined the {\em measurable} signals $y$ and $\Phi$ as
\begali{
\nonumber
 y & := L_0 \Big( {\alpha \over p+\alpha}[\bi]\Big)^\top \Omega_1 + {1\over \alpha}|\Omega_1|^2 + {1\over p+\alpha}[\Omega_2^\top \Omega_1]\\
\nonumber
\Omega_1 & := {\alpha \over p + \alpha} [\bv - R\bi - L_q p\bi] \\
\nonumber
 \Omega_2 & := \Omega_1 - L_0 {\alpha p \over p + \alpha}[\bi]\\
 \Phi & := \Omega_1 + \Omega_2.
 \lab{yphi}
}
\end{lemma}

\begin{proof}
\rm
From \eqref{xsqu}, and after some lengthy but straightforward calculations, it is possible to prove that
$$
L_0 \bi^\top \bx = |\bx|^2 - \psi_m^2 - \psi_m L_0 i_d,
$$
with the current in the synchronous $dq$-frame defined as
$$
\begmat{i_d \\ i_q} =  \begmat{\cos(\theta) & \sin(\theta) \\ -\sin(\theta) & \cos(\theta)}\bi.
$$
Applying the filter ${\alpha p \over p+\alpha}$, we get
\begequ
\lab{l0alp}
L_0 {\alpha p \over p +\alpha}[\bi^\top \bx] = {\alpha p \over p +\alpha}[|\bx|^2] + d .
\endequ
Now, in \cite[Lemma 2]{CHOetaltie} the following variation of the well-known Swapping Lemma  \cite[Lemma 3.6.5]{SASBOD} is established. Given two smooth functions $u,v$ and a constant $\alpha >0$, we have that
$$
\begin{aligned}
{\alpha p \over p+\alpha}[uv] = & {\alpha p \over p +\alpha}[u]v + {\alpha \over p+\alpha}[u]{\alpha p \over p+\alpha}[v] \\
& - {1\over p+ \alpha}\Big[  {\alpha p \over p+\alpha}[u] {\alpha p \over p+\alpha}[v]
\Big].
\end{aligned}
$$
The proof of Lemma \ref{lem1} is completed applying the identity above to \eqref{l0alp} and rearranging terms.
\qed
\end{proof}

It is important to underline that the tuning parameter $\alpha$ determines the bandwidth  of the filters used for the observer design, with a larger value corresponding to faster transient responses.
%
\section{Main Result}
\lab{sec3}
%
In this section we present our flux/position observer that, besides Assumptions \ref{ass1} and \ref{ass2}, requires the following standard PE assumption \cite[Subsection 2.5]{SASBOD}.

\begin{assumption} \em
\lab{ass3}
$\Phi$ is PE. That is, there exist $\delta >0$ and $T>0$ such that
$$
\int_t^{t+T} \Phi(s) \Phi^\top(s) ds \ge \delta I_2, \quad \forall t\ge 0.
$$
\end{assumption}
\begin{proposition}
\rm\label{pro1}
Consider the electrical dynamics of the IPMSM \eqref{ipmsm1}, \eqref{ipmsm2} verifying Assumptions \ref{ass1} and \ref{ass2}, with the signal $\Phi$, defined in \eqref{yphi}, fulfilling the PE Assumption \ref{ass3}. Define the active flux/position observer
\begequ
\label{observer}
\begin{aligned}
\dot{\hat{\bf \lambda}} & =\bv - R\bi  + \gamma \Phi \Bigg(y - \Phi^\top \hat\bx + \ell {\alpha p \over p+ \alpha} \big[\bi^\top {\sigma(\hat\bx)} \big]\Bigg) \\
\hat\bx & = \hat \lambda - L_q \bi \\
\\
\hat \theta & =  { {\rm atan2} (\hat{\bx}_2, \hat{\bx}_1), }
\end{aligned}
\endequ
with $y$ defined in \eqref{yphi}, the mapping
$$
\sigma(\hat\bx) = \left\{
\begin{aligned}
& ~~{\hat\bx \over |\hat\bx|}  & \qquad \text{if~~} |\hat\bx| \ge \epsilon>0 \\
& ~~\col(0,0) & \qquad \text{otherwise,}
\end{aligned}
\right.
$$
and $\epsilon\in(0,x_{\min})$ and $\gamma>0$ {\em tuning} parameters. There exist $\alpha_{\max}>0$ and $\gamma_{\max}>0$ such that for all $\alpha \leq \alpha_{\max}$ and  $\gamma \leq \gamma_{\max}$ we have
$$
|\tilde\lambda(t)| \leq m_0e^{-\rho_0 t}|\tilde\lambda(0)|, \quad \forall t \geq 0,
$$
for some ${m_0>0,\rho_0>0}$. Moreover, $\lim_{t\to\infty} |\tilde \theta(t)|=0$ exponentially fast.
\end{proposition}
\noindent {\bf Caveat} The projection is a technical modification needed to avoid a potential division by zero in the computation of $\dot{\hat{\bf \lambda}}$, which allows us to give an analytic expression for the term  $\bw(\bi, \bx,\tilde\bx)$ and to rigorously prove its boundedness. It should be pointed out that, neither in the simulations nor in the experiments of Section \ref{sec5}, it was necessary to include the projection operator.
\begin{proof}
\rm
Due to the discontinuity in the first equation of \eqref{observer}, the subsequent analysis is based on the Filippov's solution concept \cite{FIL}. In the following analysis, we simply select $\epsilon =\hal x_{\min}$. Defining the estimate of $d$ as
$$
  \Hat{d} := -\ell {\alpha p \over p+ \alpha} \big[\bi^\top {\hat\bx \over |\hat\bx|}\big],
$$
and denoting  the estimation error as $\tilde{d} := \hat{d} -d$, we get the observation error dynamics as
\begali{
\dot{\tilde\bx}& =  - \gamma \Phi (\Phi^\top \tilde\bx + \tilde{d}).
\lab{dottilx}
}
Now, notice that
\begali{
\nonumber
\tilde d &= -\ell {\alpha p \over p+ \alpha} \Big[ \bi^\top \Big( {\hat\bx \over |\hat\bx|} - {\bx \over |\bx|}\Big) \Big]\\
\lab{tild}
& = - {\alpha p \over p+ \alpha} [\bw^\top (\bi,\bx,\tilde\bx) \tilde\bx],
}
where we have moved the constant $\ell$ inside the filter and defined the continuous mapping $\bw:\rea^2 \times \rea^2 \times \rea^2 \to \rea^{2}$. The existence of this factorization is ensured invoking the Lagrange reminder representation of the Taylor series expansion and noticing that the term in parenthesis in the first equation of \eqref{tild} is zero at $\tilde \bx=0$. We may get the expression of $\bw$ as follows. 

{\em Case 1.} Indeed, if $|\hat \bx| \ge \hal x_{\min}$, we have
$$
\begin{aligned}
{1\over \ell}\bw^\top(\bi,\bx,\tilde{\bx})\tilde \bx
 := & \bi^\top \Big( {\hat\bx \over |\hat\bx|} - {\bx \over |\bx|}\Big) \\
 = & \bi^\top \Bigg[{1\over |\hat\bx|}\tilde\bx
 + {|\bx| - |\hat\bx| \over |\hat\bx||\bx|} \bx 
 \Bigg] \\
  = & \bi^\top \Bigg[{1\over |\hat\bx|}\tilde\bx
 + {|\bx|^2 - |\hat\bx|^2 \over |\hat\bx||\bx|(|\bx|+|\hat\bx|)} \bx 
 \Bigg] \\
  = & \bi^\top \Bigg[{1\over |\hat\bx|}\tilde\bx
 - {\tilde{\bx}^\top (\bx + \hat\bx) \over |\hat\bx||\bx|(|\bx|+|\hat\bx|)} \bx 
 \Bigg] \\
 = & {1\over |\hat\bx|} \tilde{\bx}^\top \bi
 - \tilde{\bx}^\top {(\bx + \hat\bx)\over |\hat\bx||\bx|(|\bx|+|\hat\bx|)} \bi^\top \bx
 \\
 = & \tilde{\bx}^\top \Bigg[
 {\bi \over |\hat\bx| }-  {(\bx + \hat\bx)\over |\hat\bx||\bx|(|\bx|+|\hat\bx|)} \bi^\top \bx
 \Bigg]\\
 = & \tilde{\bx}^\top {1\over |\hat \bx|} \Bigg[
 I_{2}
 -    {(\bx + \hat\bx) \bx^\top\over |\bx|(|\bx|+|\hat\bx|)}
 \Bigg] \bi.
\end{aligned}
$$
Therefore, we have
$$
\bw(\bi,\bx,\tilde\bx)
=
{\ell \over |\hat \bx|} \Bigg[ I_{2}
 -    { (\bx + \hat\bx) \bx^\top \over |\bx|(|\bx|+|\hat\bx|)} \Bigg] \bi,
$$
with $\hat\bx = \tilde\bx + \bx$. It is clear that 
\begequ
\label{ineq:w}
|\bw(\bi,\bx,\tilde{\bx})| \le {\ell \over |\hat\bx|} \|\bi\|_\infty \le {4\ell\over x_{\min}} \|\bi\|_\infty,
\endequ
invoking the fact
$
\left|{(\bx + \hat\bx) \bx^\top\over |\bx|(|\bx|+|\hat\bx|)}\right| \le 1.
$

{\em Case 2.} For the other case of $|\hat{\bx}| < \hal x_{\min} $, we have
$$
{1\over \ell}\bw^\top (\bi, \bx,\tilde\bx) \tilde\bx = - \bi^\top {\bx\over|\bx|}.
$$
It yields
$$
\bw(\bi, \bx,\tilde\bx) = - \ell{\bi^\top \bx \over |\bx| |\tilde\bx|^2} \tilde\bx,
$$
which is well-defined since $|\bx| \ge x_{\min}$ and 
$
|\tilde\bx| \ge |\bx| - |\hat \bx| \ge  x_{\min} - \hal x_{\min} >0.
$
For this case, we have
\begin{equation}
\label{ineq:w2}
 |\bw(\bi,\bx,\tilde\bx)| \le {2\ell \over x_{\min}} \|\bi\|_\infty.
\end{equation}
Together with \eqref{ineq:w}, we get the lower bound of $|\bw|$ for all $\hat\bx \in \rea^2$.

The equation \eqref{tild} admits a state-space realization
\begali{
\nonumber
\dot{z} & = -\alpha z + \alpha^2 \bw^\top \tilde{\bx} \\
\lab{starea}
\tilde d & = z - \alpha \bw^\top \tilde{\bx}.
}
Choosing, without loss of generality, $\alpha=c_0 \gamma$ for some $c_0>0$, replacing \eqref{starea} in \eqref{dottilx} we can write the error equation as
\begequ
\label{dyn1}
\dot \chi = \gamma {A}_0(t) \chi + \gamma^2  \begmat{ -c_0 \Phi \bw^\top \\ \\  c^2_0 \bw^\top  }\tilde{\bx},
\endequ
where we defined
$$
\chi:=\begmat{\tilde{\bx} \\ \\ z},\;A_0(t):= \begmat{- \Phi(t) \Phi^\top(t) & -\Phi(t) \\ \\ 0 & -c_0}.
$$
It is well-known \cite{ANDetal,SASBOD} that, under Assumptions \ref{ass2} and \ref{ass3}, the unperturbed part of \eqref{dyn1}, namely, the linear time-varying system $\dot \chi = \gamma A_0(t) \chi $, is GES. Equivalently,
\begin{equation}
\label{convg_t}
 |\chi(t)| \le m e^{-\rho t}|\chi(0)|,
\end{equation}
for some positive constants $m$ and $\rho$. Moreover, it is shown in \cite[Corollary 2.1, pp. 49]{ANDetal} that there exists $\gamma_{\max}>0$ such that for all $\gamma \in (0,\gamma_{\max})$, there exists positive constants $m_1$ and $\rho_1$---{\em determined} by $T$ and $\delta$ in the PE condition of Assumption \ref{ass3}---such that
\begalis{
\rho & = \rho_1 \gamma + \mathcal{O}(\gamma^2)\\
 m & = m_1 + \mathcal{O}(\gamma)
}
where $\mathcal{O}(\varepsilon)$ is the uniform ``big $\mathcal{O}$" symbol.\footnote{That is $f(z,\epsilon)=\mathcal{O}(\epsilon)\;\Leftrightarrow\;|f(z,\epsilon)|\leq C \epsilon$  with $C$ independent of $z$ and $\epsilon$.} This means, that there exists positive constants $\rho_2,~ m_2$---independent of $\gamma$---guaranteeing
\begalis{
\rho_1\gamma - \rho_2 \gamma^2 & \le \rho \le \rho_1 \gamma + \rho_2 \gamma^2\\
m &\le m_1 + m_2\gamma.
}
Let us introduce now the {\em time scale change} ${d\tau \over dt}=\gamma$ and write the error equation \eqref{dyn1} in the $\tau$ time scale as
\begequ
\label{dyntau}
{
\chi' = {A}(\tau) \chi + \gamma  \begmat{ -c_0 \Phi \bw^\top \\ \\  c^2_0 \bw^\top  }\tilde{\bx},}
\endequ
where we defined
$A(\tau) := A_0({\tau \over \gamma})$ and, to avoid cluttering, we introduce the notation $(\cdot)':={d \over d\tau}(\cdot)$. Now we rewrite \eqref{convg_t} in the $\tau$-time scale and compute the bounds
\begali{
\nonumber
|\chi(\tau)| & \le  (m_1 + m_2 \gamma) e^{-(\rho_1 - \rho_2\gamma)\tau} |\chi(0)| \\
\nonumber
& \le (m_1 + m_2 \gamma_{\max}) e^{-(\rho_1-\rho_2\gamma_{\max}) \tau}|\chi(0)| \\
& =: m_* e^{ -\rho_* \tau} |\chi(0)|,
\label{convg:tau}
}
where defined
\begalis{
m_*&:= m_1 + m_2 \gamma_{\max}\\
\rho_*&:=\rho_1-\rho_2\gamma_{\max},
}
and we require that $\gamma_{\max} >0$ is small enough to ensure $\rho_*>0$. The key observation here is that $m_*$ and $\rho_*$ are \emph{independent} of $\gamma$.

We now proceed to articulate the Lyapunov stability argument for the system \eqref{dyntau} that allows us to complete the proof. Towards this end, we see that from Assumption \ref{ass2}, we have that $A(\tau)$ is bounded. Consequently, invoking \cite[Theorem 4.14]{KHA} and the inequality \eqref{convg:tau} we know that the system \eqref{dyntau} admits a Lyapunov function $V(\chi,\tau)$, verifying
\begalis{
c_1|\chi|^2 \leq V(\chi,\tau) & \leq c_2 |\chi|^2\\
\nabla_\tau V +(\nabla_\chi V)^\top A(\tau) \chi   & \leq - c_3 |\chi|^2 \\
|\nabla_\chi V| & \leq c_4 |\chi|,
}
where the constants positive $c_i,\;i=1,\dots,4$, are determined by $\rho_*$, $m_*$ and the bound on $\|A(\tau)\|_\infty$, with $\|\cdot\|_\infty$ the $\call_\infty$ norm---hence they are {\em independent} of $\gamma$. Evaluating the $\tau$ time-derivative of the Lyapunov function, along the trajectories of \eqref{dyntau} and using the bounds above, as well as the chain rule for nonsmooth systems \cite{PADSAS}, yields almost everywhere
\begalis{
 {V}' & \leq  -c_3 |\chi|^2 + \gamma  (\nabla_\chi V)^\top  \begmat{ -c_0 \Phi \bw^\top  \\ \\
  c_0^2 \bw^\top}\tilde{\bx}\\
  & \leq  -c_3 |\chi|^2 +\gamma c_4    \left\|\begmat{ - c_0\Phi \bw^\top  \\ \\   c^2_0  \bw^\top} \right\|_\infty  |\chi|^2.\\
}
Given Assumption \ref{ass2}, we can always find a small $\gamma_{\max}>0$ guaranteeing
$$
c_3 > \gamma_{\max} c_4    \left\|\begmat{ - c_0\Phi \bw^\top  \\ \\   c^2_0  \bw^\top} \right\|_\infty
$$
ensuring the GES of \eqref{dyntau} and completing the proof.\qed
\end{proof}

%
\section{The Persistence of Excitation Condition}
\lab{sec4}
%
In this section, we give some verifiable sufficient conditions that ensure the PE Assumption \ref{ass2} of $\Phi$.

\begin{proposition}
\label{pro2}\rm
Consider the dynamics \eqref{ipmsm1},  \eqref{ipmsm2}. The vector $\Phi$ is PE if any of the following conditions are satisfied:{
\begenu
  \item[{\bf C1}] the signal $(\bv - R\bi -  L_s p{\bi})$ is PE;
  \item[{\bf C2}] the position-parameterized port signal
  $$
  p [-{L_0 \over 2}\bi + (L_0 \bi^\top \bc(\theta) + \psi_m)\bc(\theta)]
  $$
  is PE;
  \item[{\bf C3}]  at standstill $(\theta = \text{const})$, the time derivative of $\bi$ is PE.
\endenu}
\end{proposition}
\begin{proof}
\rm
It can be proved that
$$
\begin{aligned}
{1\over 2}\Phi & = {\alpha \over p + \alpha}[\bv - R\bi] - L_s {\alpha p \over p+\alpha}[\bi] \\
& = {\alpha \over p+\alpha }[\bv - R\bi -  L_s p{\bi}].
\end{aligned}
$$
It is well-known that a PE signal filtered by an asymptotically stable transfer function is still PE \cite[Lemma 2.6.7]{SASBOD}, thus
$$
(\bv - R\bi -  L_s p{\bi}) \in \text{PE} \quad \Longrightarrow \quad
\Phi \in \text{PE},
$$
verifying the claim \textbf{C1}. Now, from the motor dynamics we have
\begalis{
(\bv - R\bi -  L_s p{\bi}) &=p [{\lambda} -  L_s {\bi}] \in \text{PE} \\
&= p [(L_q -  L_s)\bi + (L_0 \bi^\top \bc(\theta) + \psi_m)\bc(\theta)]\\
&=p [-{L_0 \over 2} \bi + (L_0 \bi^\top \bc(\theta) + \psi_m)\bc(\theta)],
}
which proves the second claim. To prove the third claim we carry-out the following computations with $\theta$ {\em constant}:
\begalis{
 & p [-{L_0 \over 2} \bi + (L_0 \bi^\top \bc(\theta) + \psi_m)\bc(\theta)] \\
  = &  p [-{L_0 \over 2} \bi + (L_0 \bi^\top \bc(\theta))\bc(\theta)]  \\
   = & L_0 p[ - {1\over 2} \bi  +  \bc(\theta) \bc^\top(\theta) \bi ] \\
   = & L_0
\Big(  - {I_2\over 2}   + \bc(\theta)\bc^\top(\theta)\Big)  p[\bi].
}
Furthermore, we have
$$
\det \Big(  - {I_2\over 2}   + \bc(\theta)\bc^\top(\theta)\Big) = - {1\over 4},
$$
concluding the proof. \qed
\end{proof}

We bring to the readers attention the practical relevance of Condition \textbf{C3}. It is well-known that, because of a lack of observability at zero speed \cite{KOTetal}, for slow-speed operation of the motor it is necessary to use an active, probing-signal injection method to estimate the flux \cite{NAMbook}. Condition \textbf{C3} shows that injecting a high-frequency signal to the stator currents guarantees $\Phi\in $ PE at standstill. Simulation and experimental results, presented in the next section, corroborate this fact. See \cite{CHONAM,YIetal,YItcst} for further discussion on this matter.

%
\section{Simulation and Experimental Results}
\lab{sec5}
%
\subsection{Simulations}
\lab{subsec51}
Simulations were conducted to evaluate the performance of the proposed position observer with the motor parameters listed in Table \ref{tab:1}. The IPMSM is controlled by a classical decoupling field-oriented speed regulation scheme, and the observer is not connected to the closed-loop.

\begin{table}[h] \label{tab:1}
\caption{Parameters of the IPMSM: Simulation (the first column) and Experiments (the second column)}
\renewcommand\arraystretch{1.6}
\begin{center}
\vspace{0.2cm}
\begin{tabular}{l|r|r}
\hline\hline
 Number of pole pairs &  6  & 3  \\
 {PM flux linkage constant} ($\psi_m$) [Wb] &0.11 & 0.59\\
 $d$-axis inductance ($L_d$) [mH] & 5.74 & 7.6 \\
 $q$-axis inductance ($L_q$) [mH] & 8.68 & 12.9 \\
 Stator resistance ($R$) [$\Omega$]& 0.43 & 1.059 \\
 Drive inertia [kg$\cdot\text{m}^2$] & 0.01 & $\ge$ 0.01 \\
 Rated torque [N$\cdot$m] & 6.2 & 6.2
 \\
 \hline\hline
\end{tabular}
\end{center}
\end{table}
\subsubsection{Non-zero speed behaviour}
\lab{subsubsec611}
%
We first consider the motor speeding up from $10$ rad/s to $100$ rad/s with a constant torque equal to $0.5$ N$\cdot$m, see Fig. \ref{fig:1}. The parameters and initial conditions of the observer are selected as $\alpha=20$, $\gamma =10$ and $\hat{\lambda}(0)= [0.5, 2]$. The performance of the position estimator is shown in Fig. \ref{fig:2}, where perfect tracking after a short time period is observed. Fig. \ref{fig:3} gives the trajectory of the vector $\Phi$, clearly satisfying the PE condition. A load step disturbance is considered in Fig. \ref{fig:4} at $t=0.5$ s, illustrating how the observer is invariant to this perturbation.

To evaluate the conservativeness of our main proposition we tested the case with a large $\alpha =200$, whose simulation results are given in Fig. \ref{fig:5}. As expected, a small steady-state position estimation error is observed. This stems from the fact that the perturbation term in \eqref{dyn1} cannot be dominated by the GES part for large $\alpha$. To show that the same phenomenon appears increasing $\gamma$, in Fig. \ref{fig:6} we took $\gamma=100$, where a notable performance degradation is observed.

\begin{figure}
  \centering
  \includegraphics[width=0.35\textwidth]{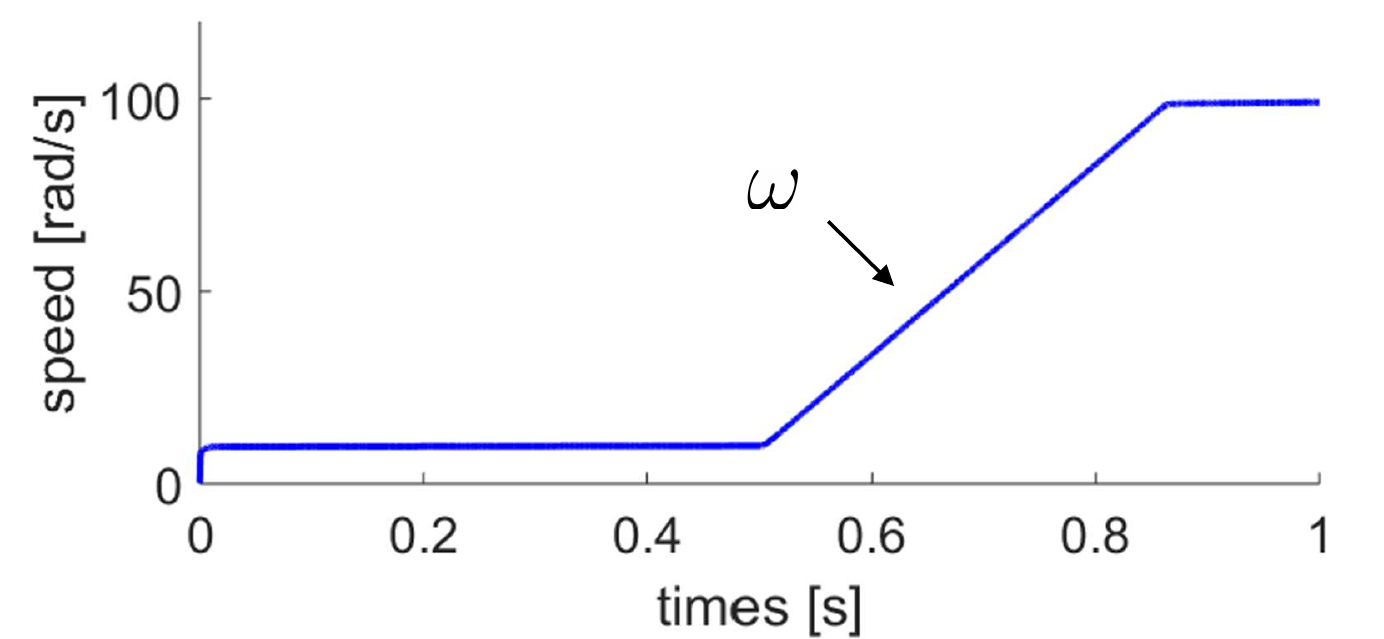}
  \caption{Reference for the rotor speed $\omega$.}\label{fig:1}
\end{figure}

\begin{figure}
  \centering
  \includegraphics[width=0.35\textwidth]{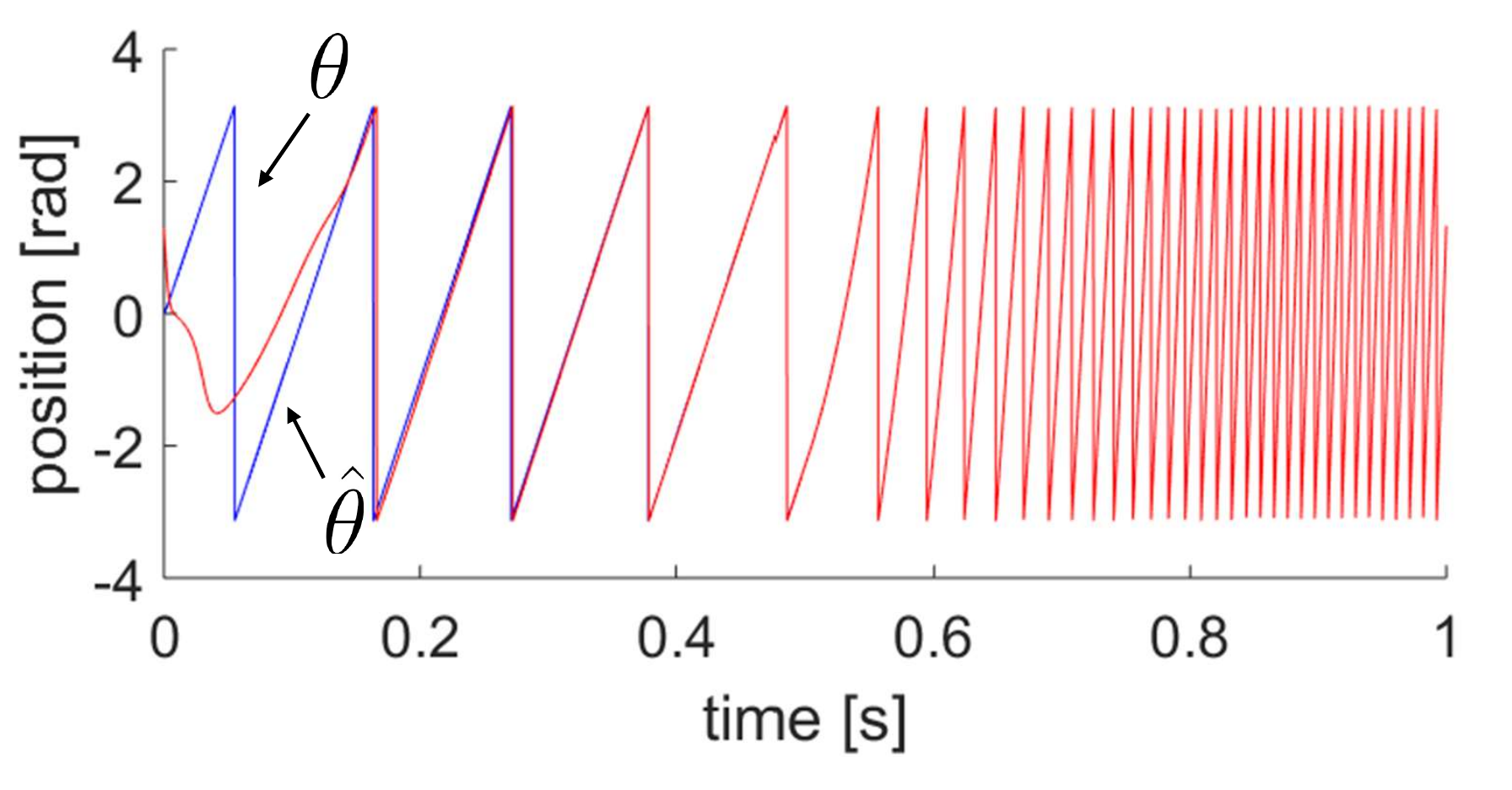}
  \caption{Angle $\theta$ and its estimate $\hat\theta$. ($\gamma=10,\; \alpha = 20$) }\label{fig:2}
\end{figure}

\begin{figure}
  \centering
  \includegraphics[width=3cm]{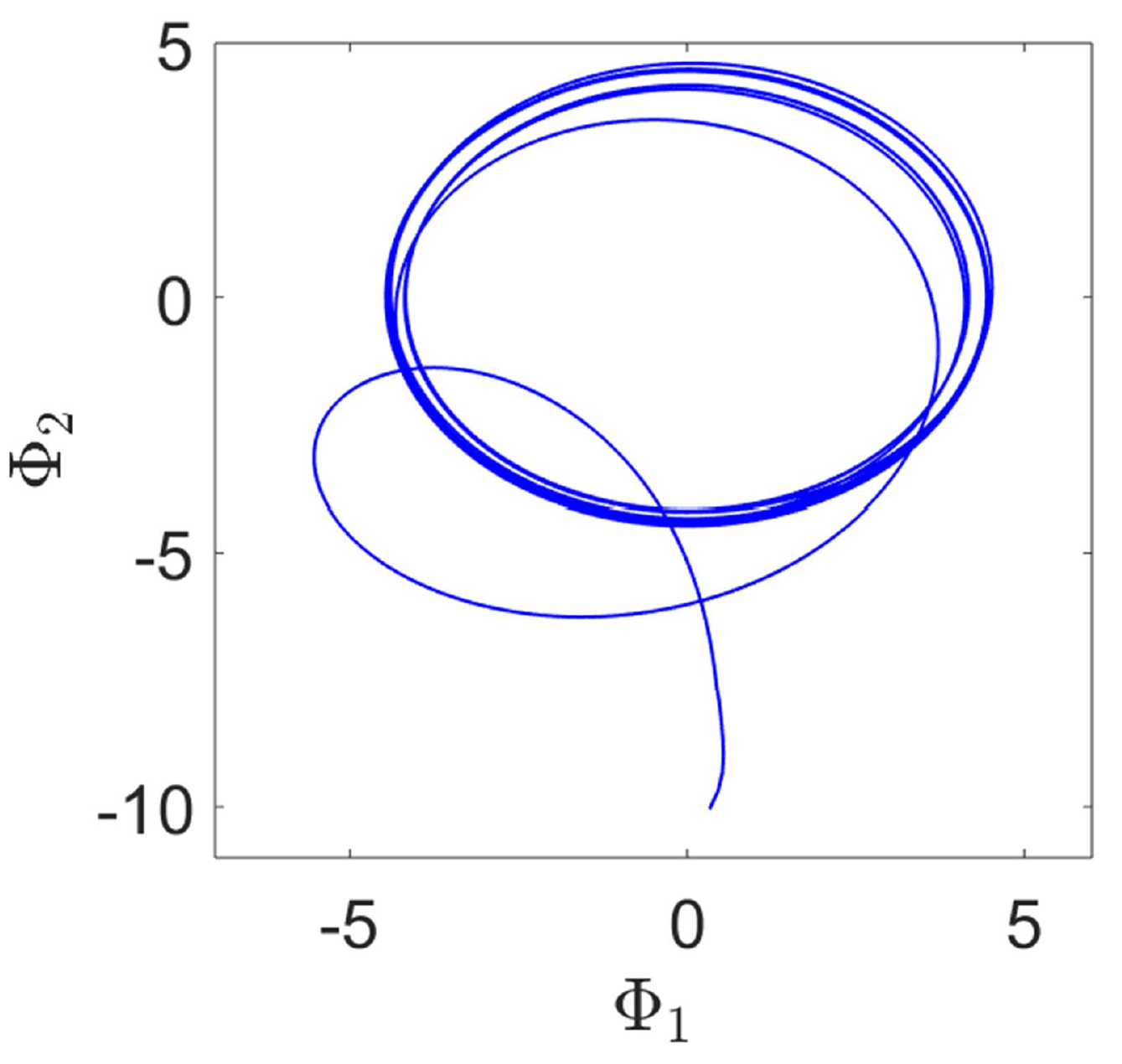}
  \caption{Trajectory of the vector $\Phi$.}\label{fig:3}
\end{figure}

\begin{figure}
  \centering
  \includegraphics[width=0.4\textwidth]{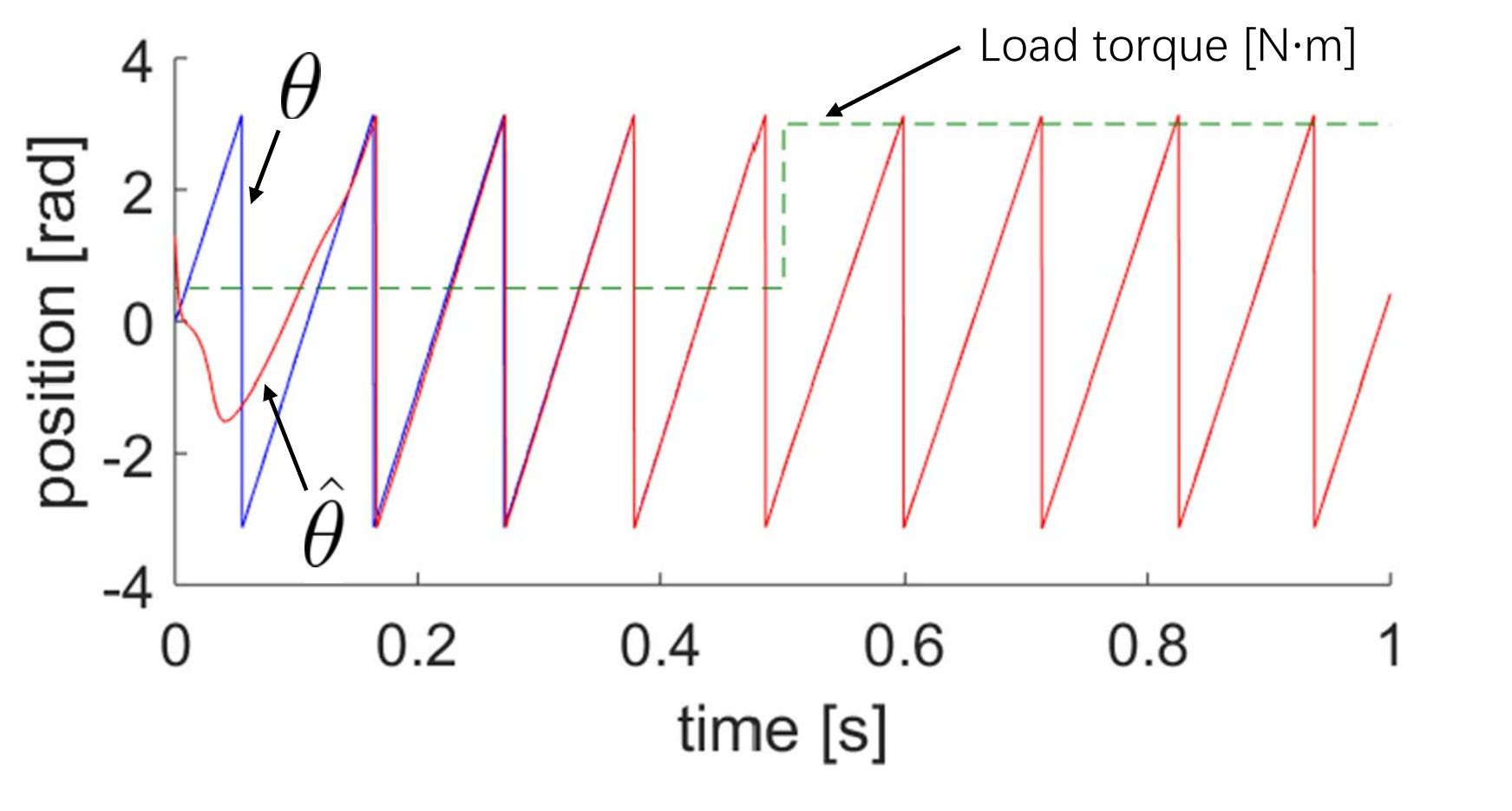}
  \caption{Angle $\theta$ and its estimate $\hat\theta$ with a load torque step.}\label{fig:4}
\end{figure}

\begin{figure}
  \centering
  \includegraphics[width=0.35\textwidth]{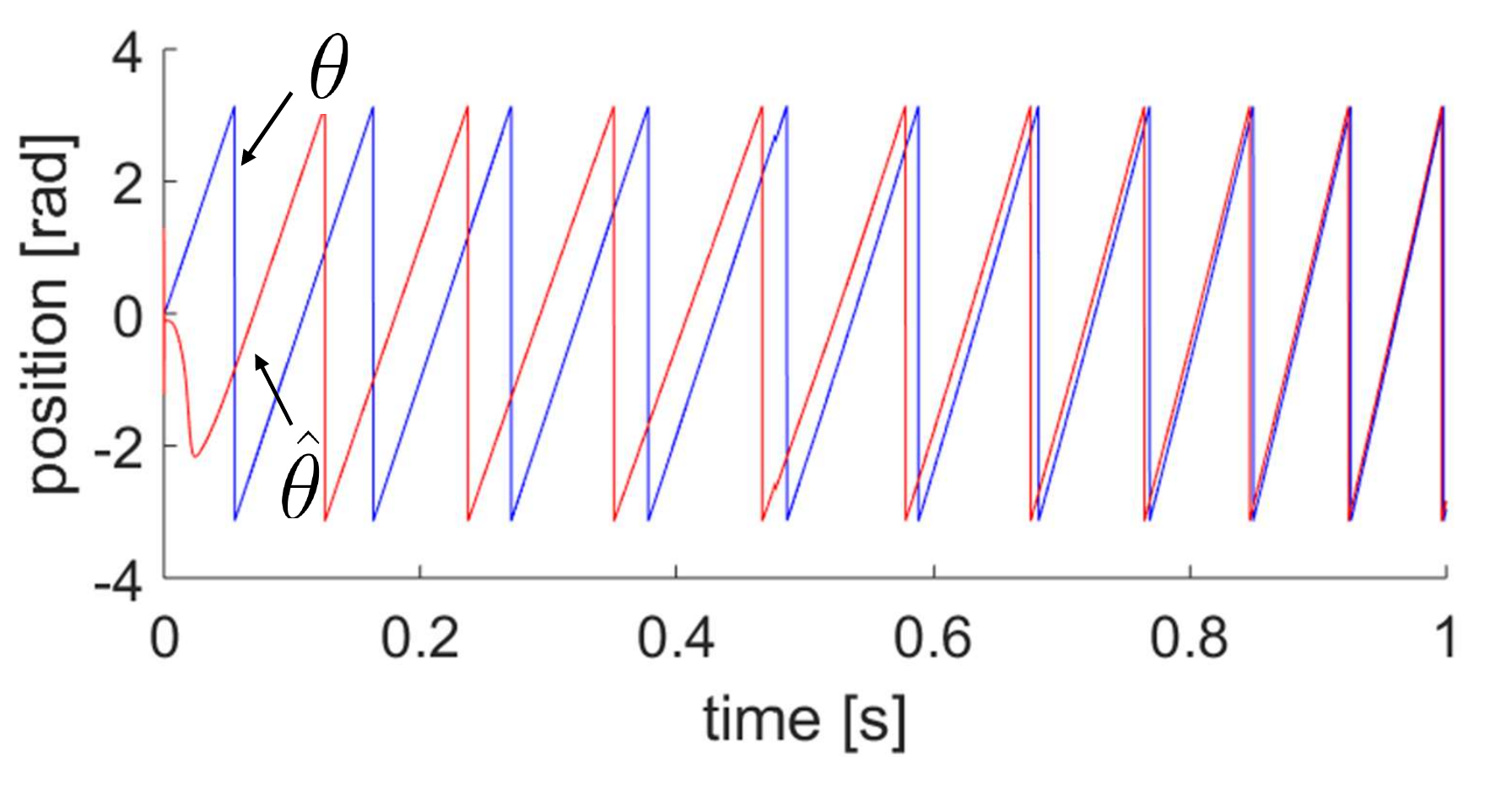}
  \caption{Angle $\theta$ and its estimate $\hat\theta$ with a large $\alpha$. ($\gamma=10,\; \alpha = 200$)}\label{fig:5}
\end{figure}

\begin{figure}
  \centering
  \includegraphics[width=0.35\textwidth]{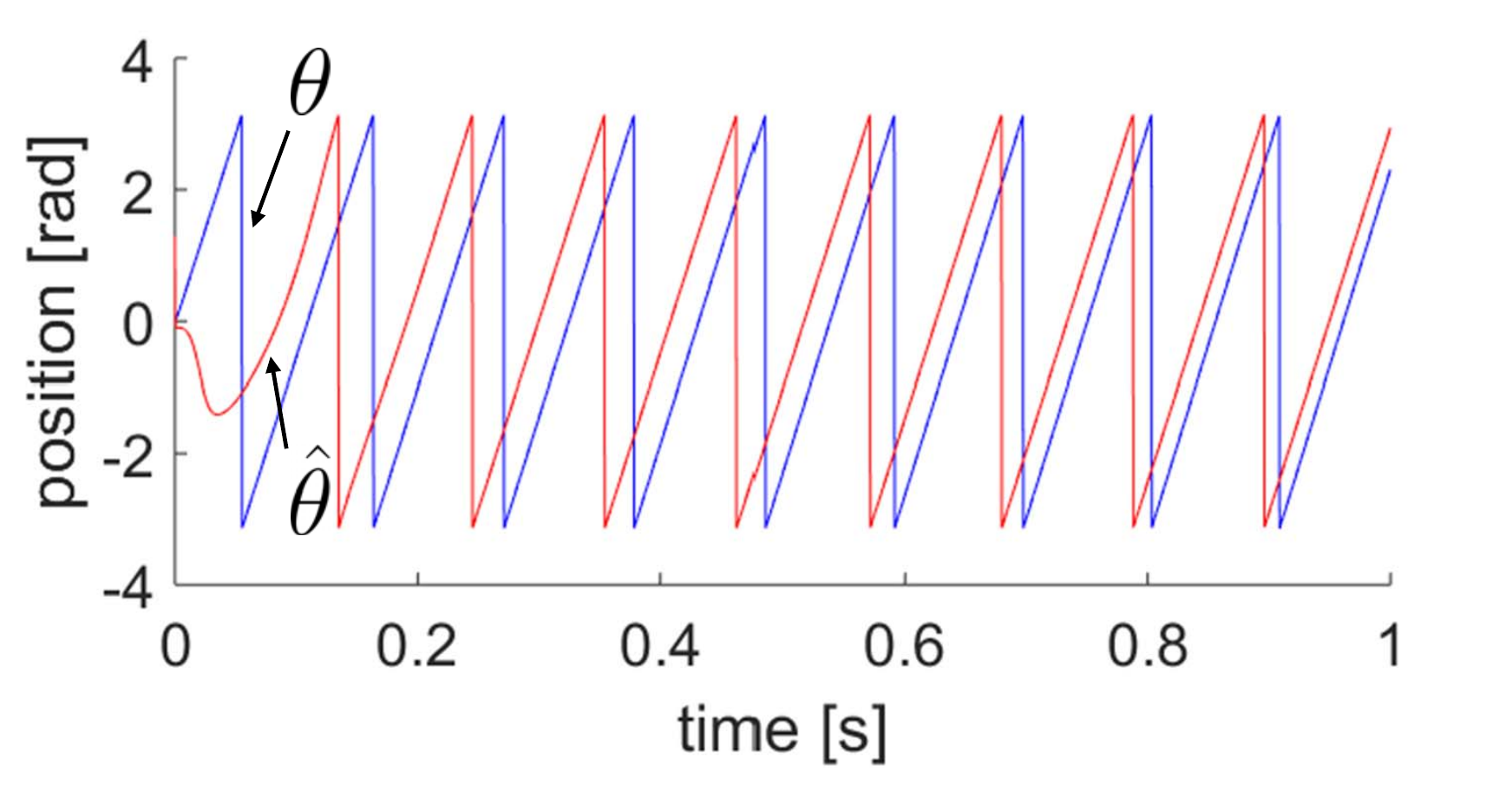}
  \caption{Angle $\theta$ and its estimate $\hat\theta$ with a large $\gamma$. ($\gamma=100,\; \alpha = 20$)}\label{fig:6}
\end{figure}
\subsubsection{Signal injection at zero speed}
\lab{subsubsec612}
%
In this subsection we simulate the motor operating at zero velocity and verify the PE condition injecting a probing signal. As indicated in condition {\bf C3} of Proposition \ref{pro2}, $\Phi$ is PE at standstill if $p[\bi]$ is PE. Fig.~\ref{fig:add} shows the behaviour of the proposed observer with or without  high-frequency injection when the motor is at a standstill.

The initial angle error was 90$^{\circ}$, that is, we used as an initial value
$
\hat{\bf \lambda} (0) = \psi_m \col(\cos(\theta + \pi/2), \sin(\theta + \pi/2)).
$
For $0\leq t\leq 0.05$ s, the estimation flux $\hat{\bx}$ stays in the initial value because $\boldsymbol{\Phi}$ is the zero vector. For $t \geq 0.05$~s, a rotating high-frequency voltage in the stationary frame is injected, as suggested in \cite{GARetal}. Then, the trajectory of $\bi$ and $\Phi$ form an ellipse as shown in Fig.~\ref{fig:add} (e) and (f), showing that $\Phi$ satisfies the PE condition, even at standstill.

\begin{figure}[tb]
  \centering
  \includegraphics[width=0.5\textwidth]{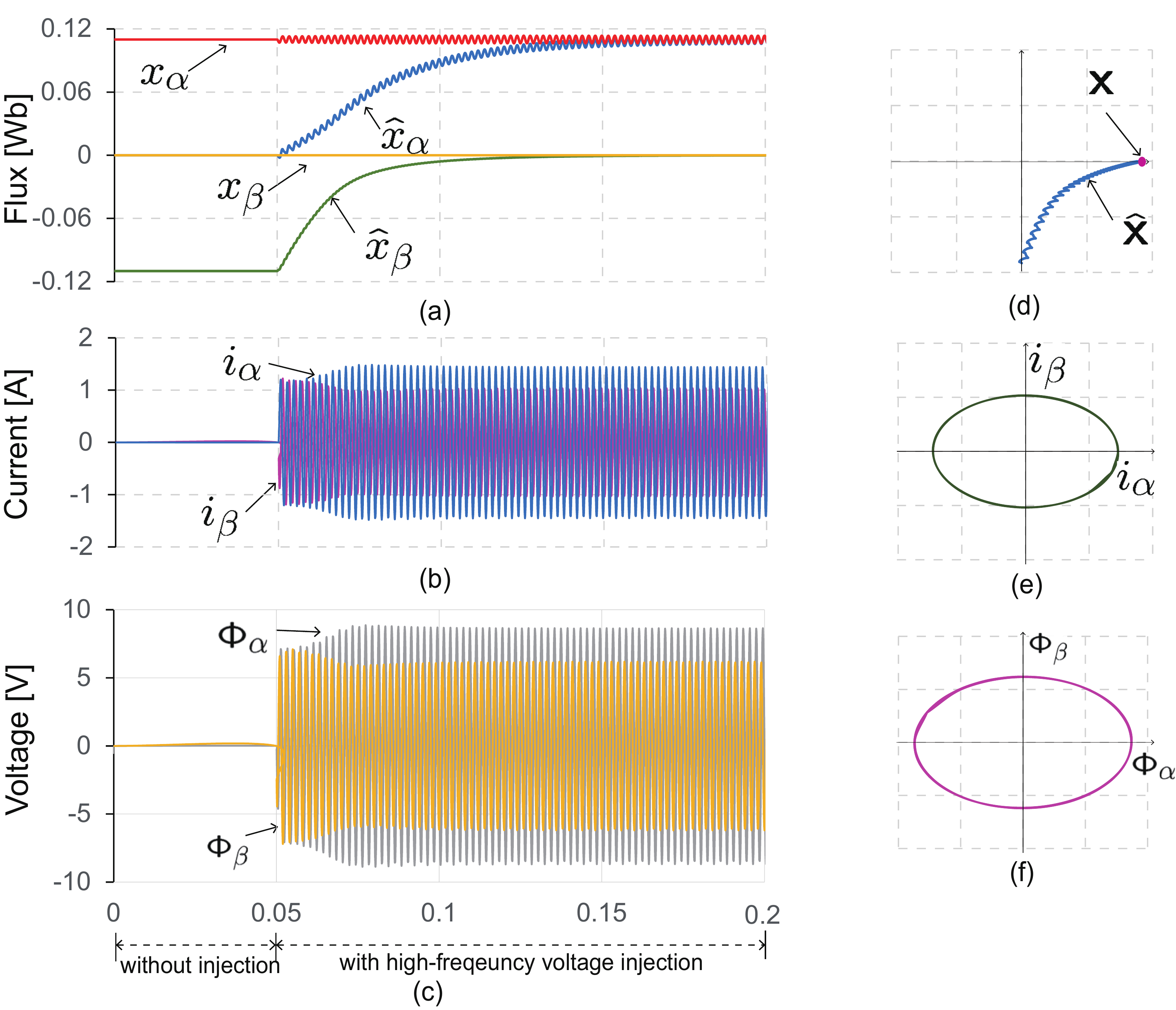}
  \caption{The performance of the observer with or without the high frequency injection at standstill: (a) active flux and its estimate; (b) stator current $\textbf{i}$; (c) regressor function $\boldsymbol{\Phi}$; (d) Lissajous curve curve of $\hat{\bx}$ and ${x}$; (e) Lissajous curve of $\textbf{i}$; (f) Lissajous curve of $\boldsymbol{\Phi}$.}
  \label{fig:add}
\end{figure}

\subsection{Experiments}
\lab{subsec62}
%
Experiments were conducted with a test bench which comprised two IPMSMs, shown in Fig. \ref{fig:test_bench}. The tested motor ran in speed control mode, using its own high-resolution shaft sensor to close the current, torque and speed loops. The other motor was coupled and ran in torque-control mode. This motor was used as the load, providing a controllable and programmable load torque. The shafts of those motors were connected via a toothed belt, which also couples the inertial wheel and secures a relatively stable speed, see Fig. \ref{fig:shaft}.

\begin{figure}[]
  \centering
  \subfigure[Experimental testing setup]{
  \includegraphics[width=0.19\textwidth]{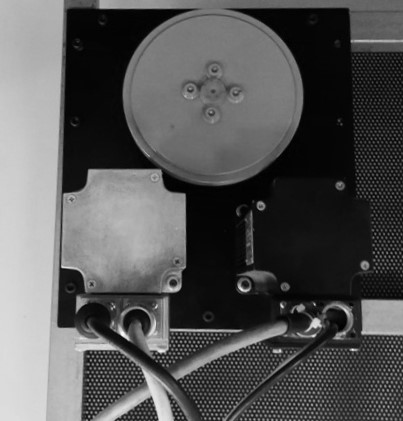}  \label{fig:test_bench}
  }
  \subfigure[Shafts connection of two motors]{
  \includegraphics[width=0.245\textwidth]{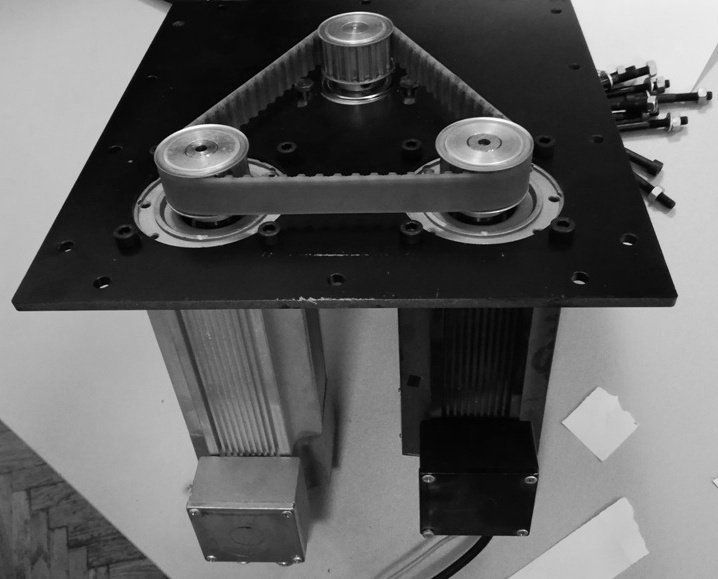}  \label{fig:shaft}
  }
\caption{Experimental testing setup and connections}
\end{figure}

The sampling time and the closed loop bandwidth of the speed loop are $T_{sw}$ = 200 $\mu$s and $f_{bw}(\omega)$ = 15 Hz, respectively. Both motors are supplied from industrial DSP-controlled inverters equipped with the real-time Ethernet link. The pulse-width modulation frequency of the IGBT inverters is set to 5 kHz, the sampling time of the digital current loop is 10 kHz, while the closed loop bandwidth of the current loop is 1.5 kHz.

\subsubsection{High speed behaviour}
\lab{subsubsec622}
%
In the first test, the motor ran at the \emph{constant} speed of about 1000 RPM, satisfying the PE condition thus without signal injection. We pulsed the torque---from zero up to the rated toque and then reducing back. Hence, the speed made small variation and settled back to the set value. We select the observer parameters as $\gamma=50$ and $\alpha=5$ with the initial condition $\hat{\lambda}(0)= [0.5~2]^\top$. Instead of measuring the load torque $T_L$ directly, we recorded the current reference $i_q^d$, which is proportional to $T_L$. The experimental results are shown in Fig. \ref{fig:test1}, which illustrates the high estimation performance at constant high speeds.

\begin{figure}[htp!]
  \centering
  \includegraphics[width=0.4\textwidth]{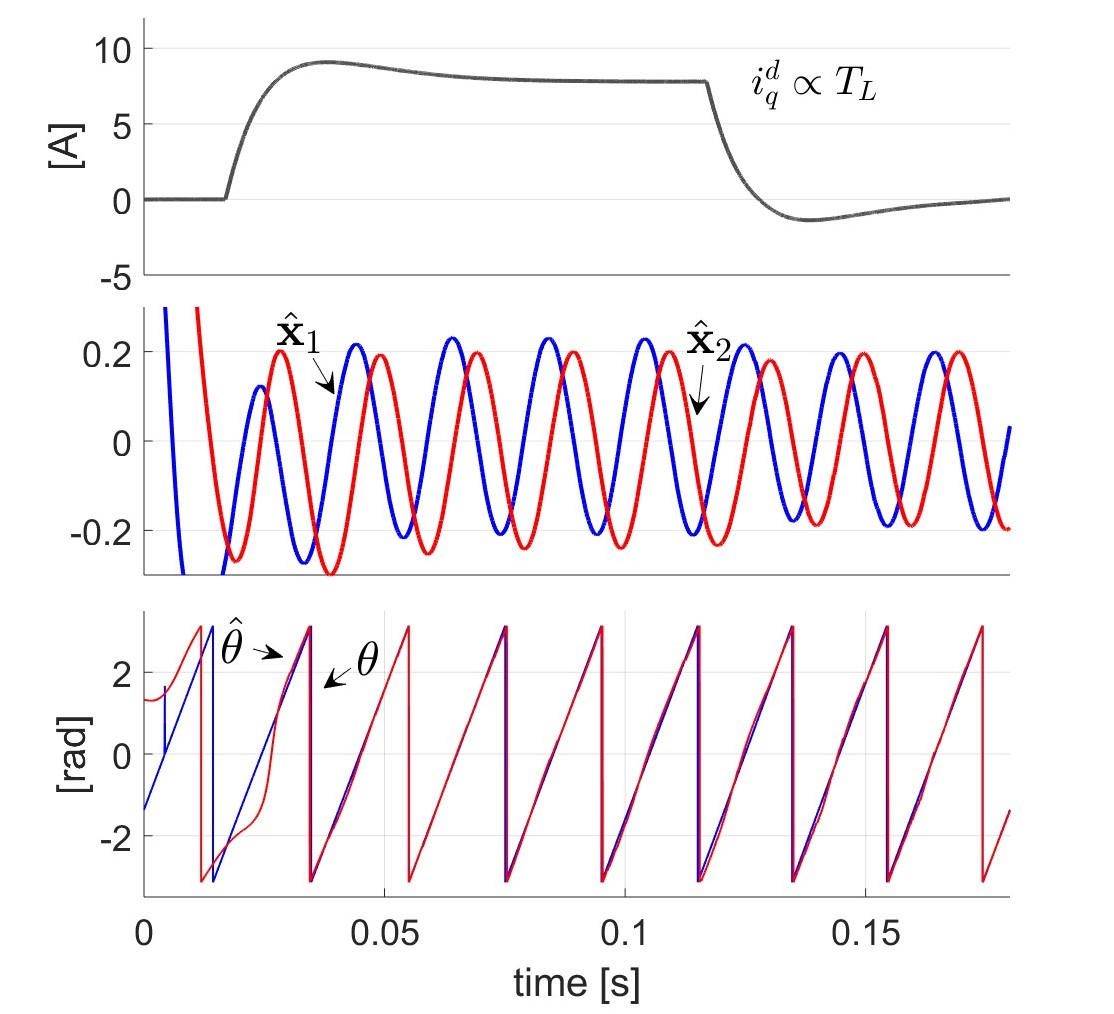}
  \caption{Angle and active flux estimates at a high constant speed (1000 RPM) under a time-varying load condition}\label{fig:test1}
\end{figure}

\subsubsection{Effect of signal injection at low speed}
\lab{subsubsec621}
%
Figs. \ref{fig:test2}-\ref{fig:test2-2} show the performance of the proposed position observer with a time-varying speed, that is, from a low-speed area to a high-speed one, then decelerating back. We select the observer parameters as $\alpha=10$ and $\gamma=50$. The difference between Fig. \ref{fig:test2} and Fig. \ref{fig:test2-2} is that in the former the motor operates {\em without} signal injection. On the other hand, in the latter figure we probed a {\em high-frequency signal} into the stator voltages of the form
\begequ
\label{signal_inj}
\bv = \bv^* +  A e^{\mathcal{J}\omega_h t}\begin{pmatrix} 1\\ 0 \end{pmatrix},
\quad
\mathcal{J} = \begin{pmatrix} 0 & -1 \\ 1 & 1 \end{pmatrix},
\endequ
where $A$ is the amplitude of the injected signals, $\omega_h$ is its frequency, and $\bv^*$ is the normal (low-frequency) control. Here, we select $A=4$ V, and $\omega_h= 2\pi\times 400$ rad/s.

In both scenarios we get excellent estimation performance at the high-speed region; however, the observer in the absence of signal injection has relatively poor performance during the intervals $[0,0.3]$ s and $[0.85,1.2]$ s---when the motor runs at low speeds. The performance improvement achieved with the signal injection in low speeds is remarkable.  We also draw the curve of currents $\bi \in \rea^2$ in Fig. \ref{fig:test2-3}. In addition, Fig. \ref{fig:test0} displays the results at a constant low speed $\omega=20$ RPM with constant load, where signal injection is applied as \eqref{signal_inj} with $A=6$ V and $\omega_h= 2\pi\times 400$ rad/s.

\begin{figure}[htp!]
  \centering
  \includegraphics[width=0.4\textwidth]{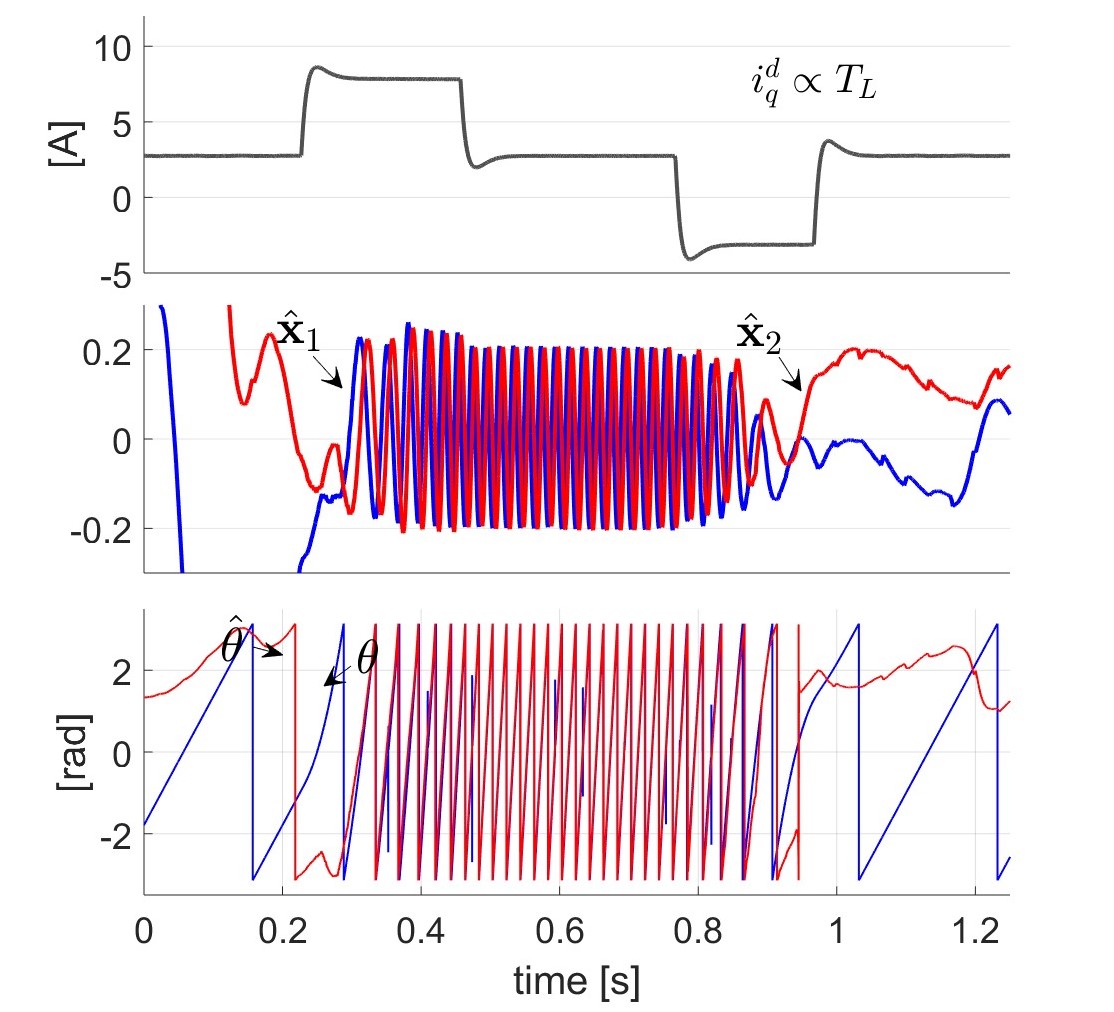}
  \caption{Angle and active flux estimates with accelerated speed (from 100 to 1000 RPM) \emph{in the absence of} high-frequency signal injection}\label{fig:test2}
\end{figure}

\begin{figure}[htp!]
  \centering
  \includegraphics[width=0.4\textwidth]{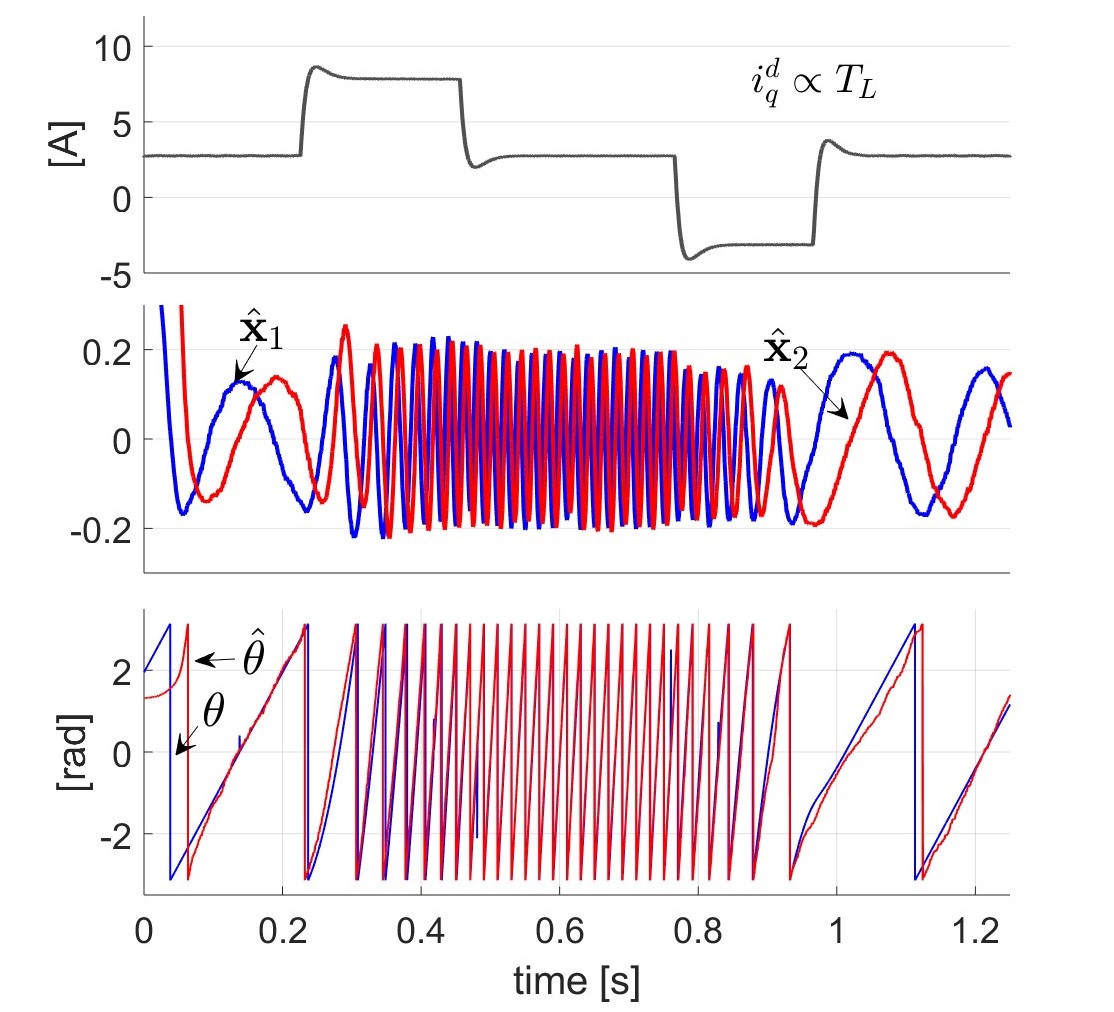}
  \caption{Angle and active flux estimates with accelerated speed (from 100 to 1000 RPM) \emph{with} high-frequency signal injection}\label{fig:test2-2}
\end{figure}

\begin{figure}[htp!]
  \centering
  \includegraphics[width=0.4\textwidth]{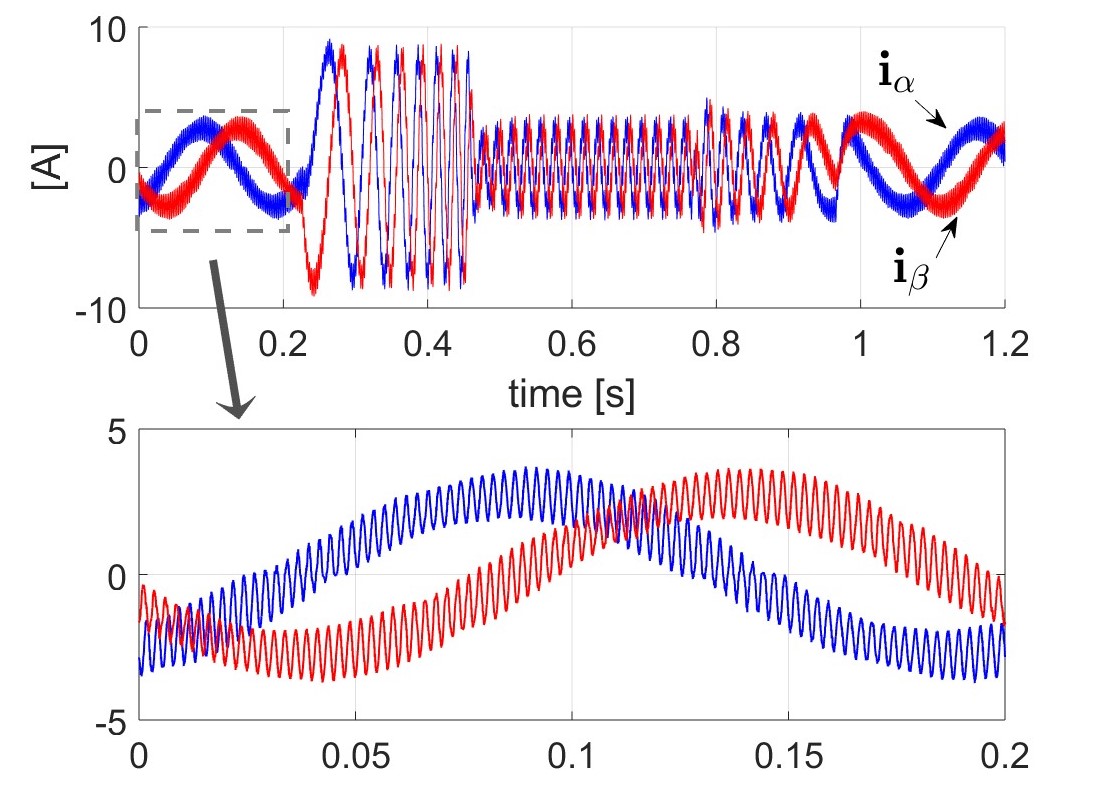}
  \caption{The currents signal $\bi$ corresponding to Fig. \ref{fig:test2-2}.}\label{fig:test2-3}
\end{figure}

\begin{figure}[htp!]
  \centering
  \includegraphics[width=0.4\textwidth]{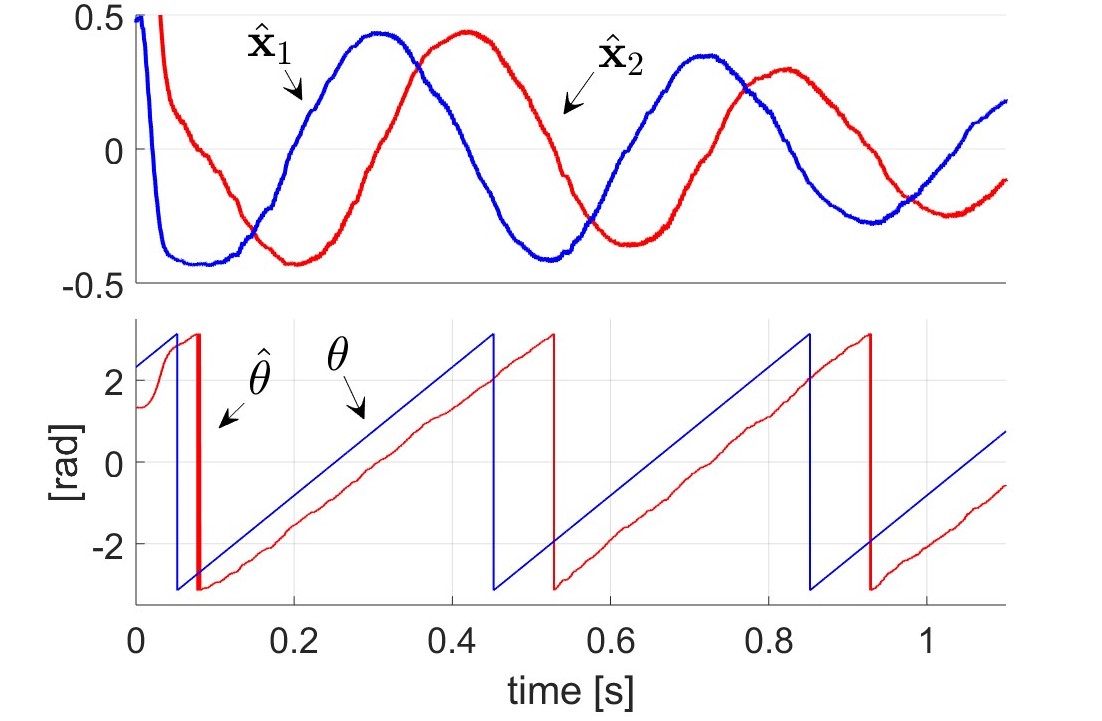}
  \caption{Angle and active flux estimates at a high constant speed (20 RPM) under constant load condition}\label{fig:test0}
\end{figure}

We also test the case of speed reversal. Fig. \ref{fig:testA-4} illustrates the estimation performance while speed is reversing from 20 RPM to -20 RPM in the presence of signal injection, verifying that the proposed position observer ensures satisfactory performance at low speeds and standstill.

\begin{figure}[htp!]
  \centering
  \includegraphics[width=0.4\textwidth]{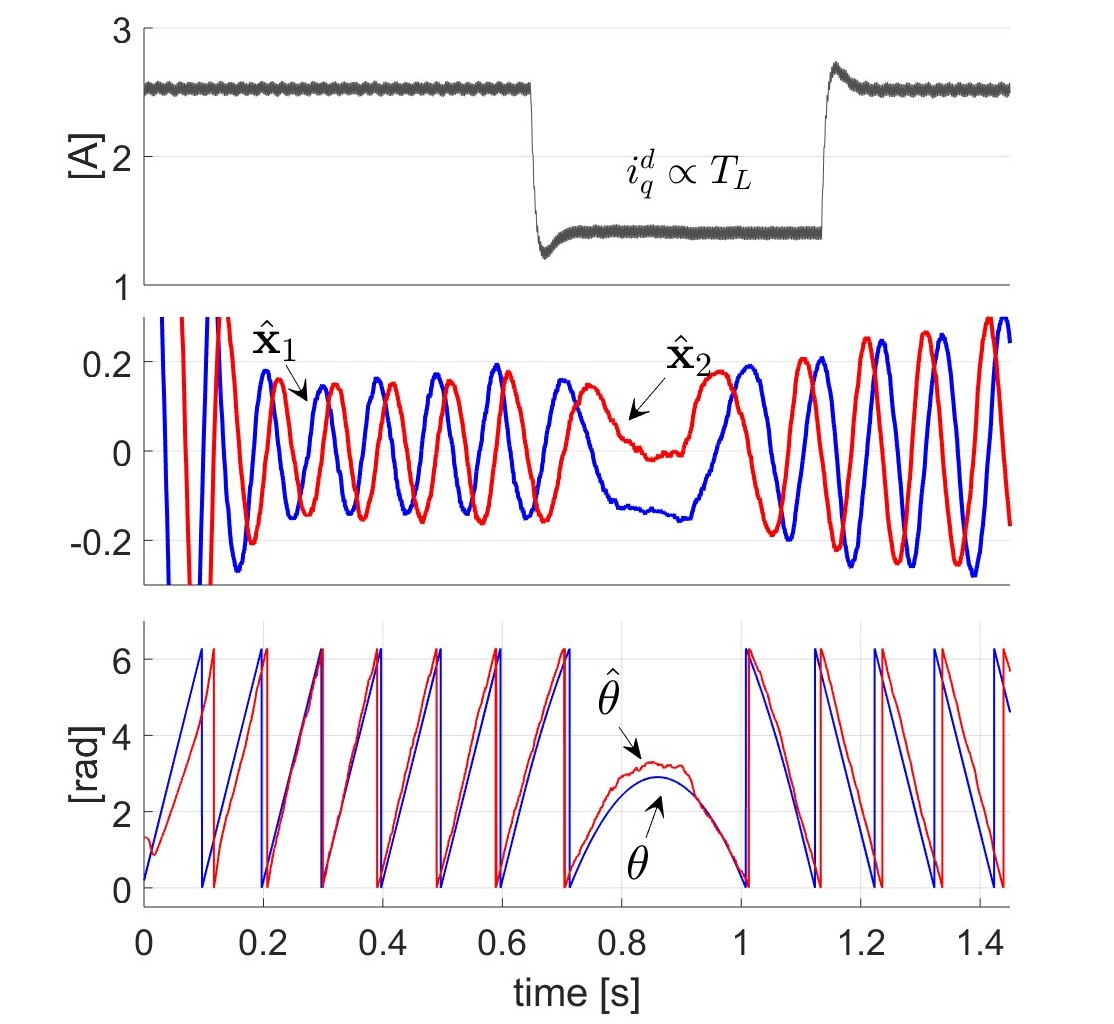}
  \caption{Angle and active flux estimates with speed reversal under pulse load condition}\label{fig:testA-4}
\end{figure}

\subsubsection{Effect of tuning parameters}
\lab{subsubsec623}
%
In the second group of experiments, we study the effect of the observer tuning parameters $\gamma$ and $\alpha$. Figs. \ref{fig:test_B-1} and \ref{fig:test_B-2} show the performance of the observer with different $\gamma$ and $\alpha$, respectively. In both cases we used the same operation mode of the one in Fig. \ref{fig:test1}. We observe that the convergence speed decreases with $\gamma$ reducing from $\gamma=50$ to $\gamma=20$, coinciding with the theoretical analysis. That is, \emph{sufficient small} $\gamma>0$ will yield slow convergence speed. Since we cannot guarantee the exponential stability with an extremely large $\gamma$, the performance degradation can be observed in Fig. \ref{fig:test_B-1} with $\gamma=200$.

\begin{figure}[htp!]
  \centering
  \includegraphics[width=0.4\textwidth]{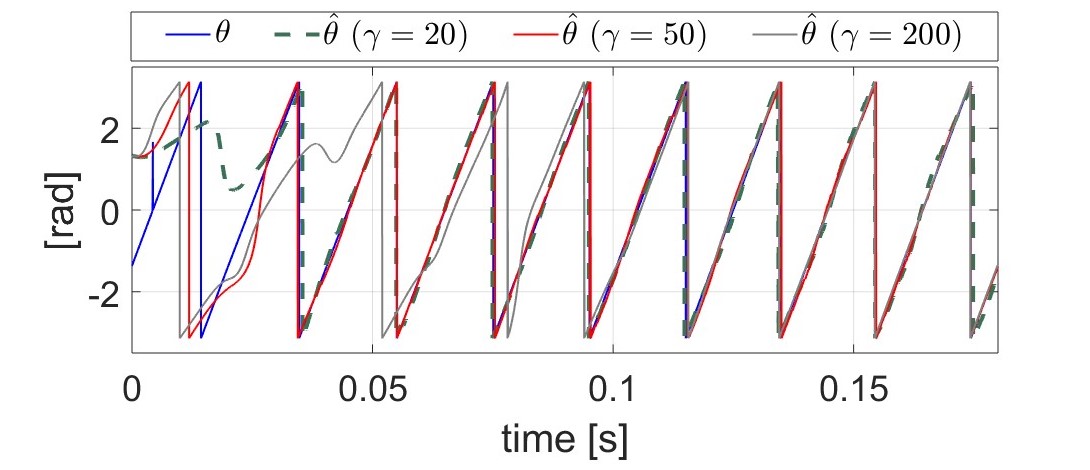}
  \caption{Angle estimates at 1000 RPM with different $\gamma$}\label{fig:test_B-1}
\end{figure}

We also present the effects from the tunable parameter $\alpha$ in Fig. \ref{fig:test_B-2}. The convergence speeds are slowed down with $\alpha$ decreasing from 5 to 2. However, a large $\alpha=10$ has some deleterious effects apparently due to increase in the signal-to-noise ratio.

\begin{figure}[htp!]
  \centering
  \includegraphics[width=0.4\textwidth]{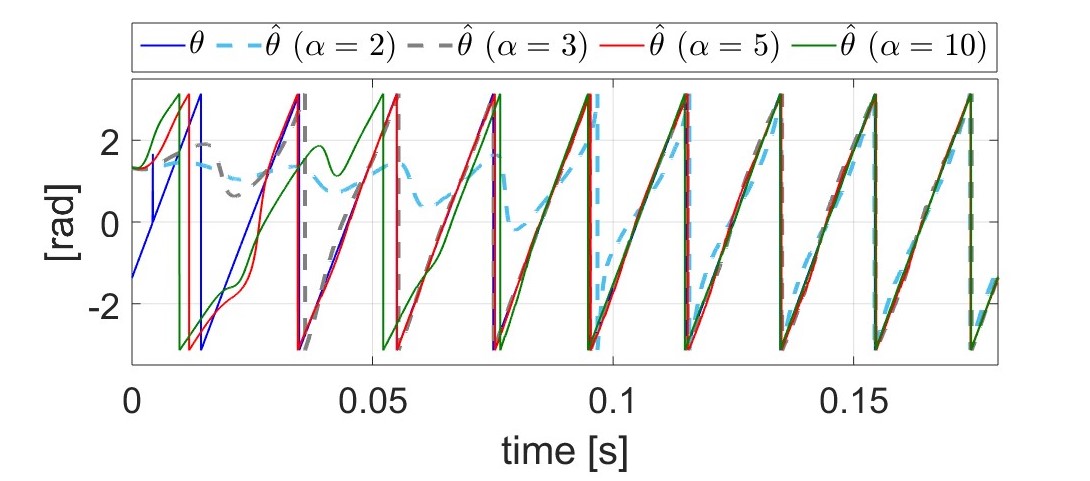}
  \caption{Angle estimates at 1000 RPM with different $\alpha$}\label{fig:test_B-2}
\end{figure}

\subsubsection{Robustness to uncertainty in the motor parameters}
\lab{subsubsec624}
%
In this subsubsection, we evaluate the sensitivity to the motor parameters of the proposed observer. In Fig. \ref{fig:test_C-1}, we use in the observer the nominal value of $R$ from Table \ref{tab:1}, as well as twice its value. It can be observed that it is almost impossible to observe any difference. We roughly conclude that the new observer design is very robust to the uncertainty of stator winding resistance $R$. Fig. \ref{fig:test_C-2} shows the robustness to the uncertainty from $L_q$ and $L_d$, where we used in the observer twice the value of $L_q$ and $L_d$ from the Table \ref{tab:1}. Despite the inaccuracy introduced in observer design, only a slight degradation is observed, thus we assume a high robustness to inductances of the proposed design. Fig. \ref{fig:test_C-3} illustrates the robustness to the PM flux linkage constant $\psi_m$, where we use the 1.5 times larger than its normal value. It has the most conspicuous performance degradation compared with other cases, though its performance is still acceptable.

\begin{figure}[htp!]
  \centering
  \includegraphics[width=0.4\textwidth]{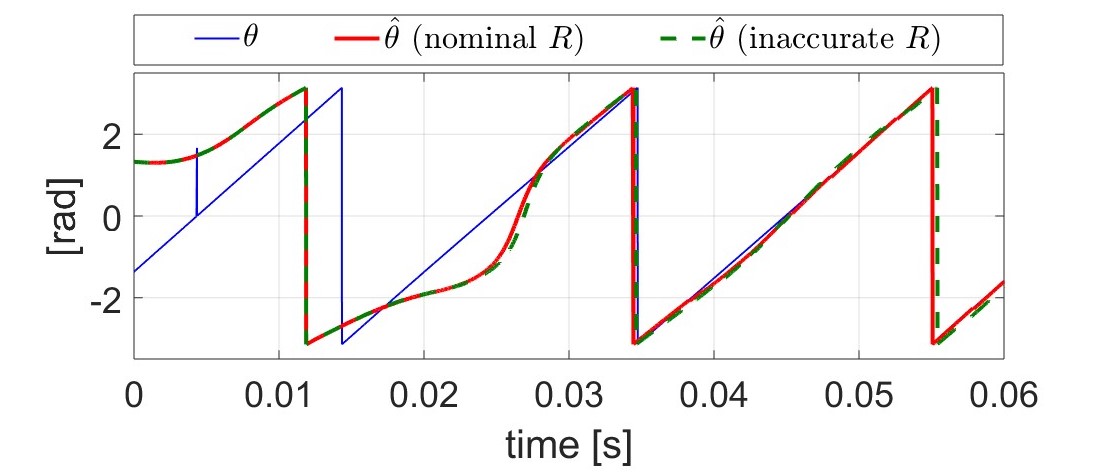}
  \caption{Sensitivity to the stator winding resistance $R$}\label{fig:test_C-1}
\end{figure}

\begin{figure}[htp!]
  \centering
  \includegraphics[width=0.4\textwidth]{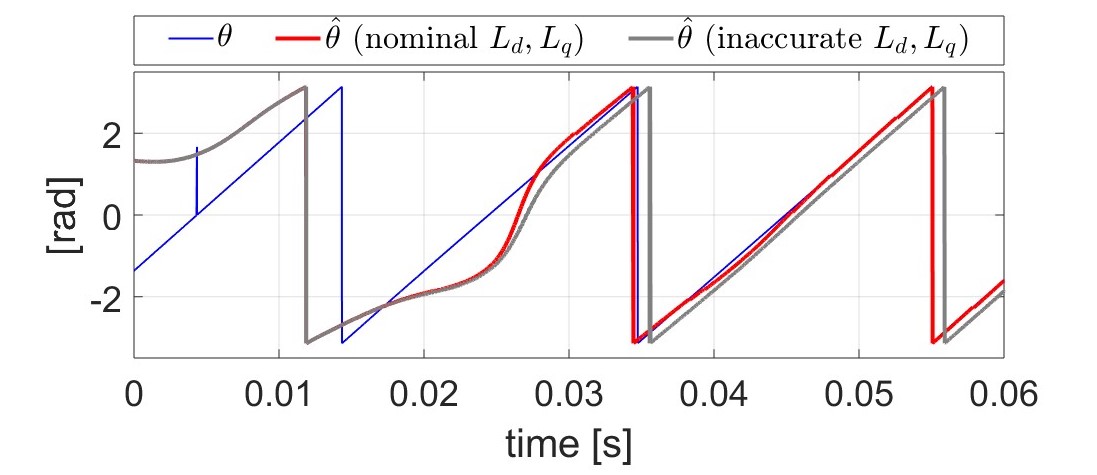}
  \caption{Sensitivity to the inductances $L_d$ and $L_q$}\label{fig:test_C-2}
\end{figure}

\begin{figure}[htp!]
  \centering
  \includegraphics[width=0.4\textwidth]{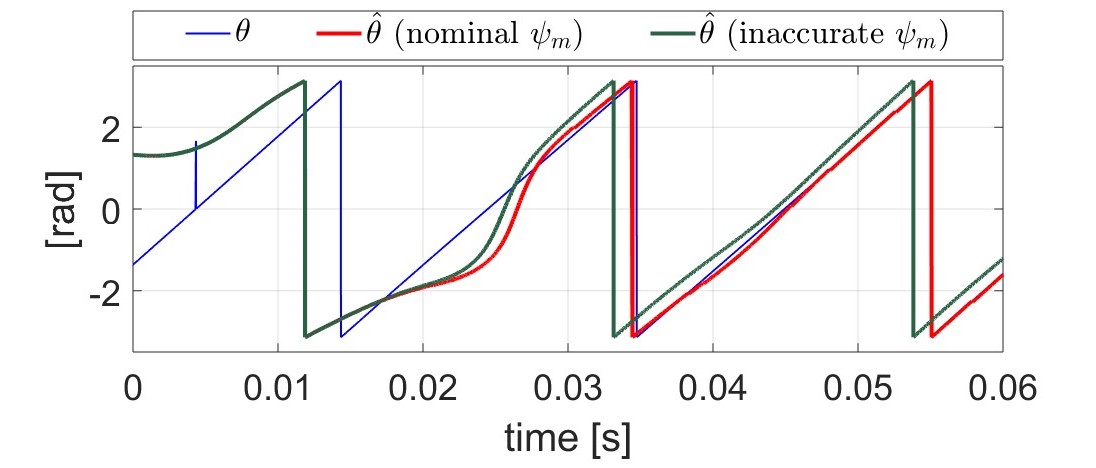}
  \caption{Sensitivity to the PM flux linkage constant $\psi_m$}\label{fig:test_C-3}
\end{figure}
\subsubsection{Comparison with other signal injection methods}
\lab{subsubsec625}
%
Finally, we compare the proposed observer with some other signal injection-based design at low speeds. In the last two decades, signal injection-based methods have been widely applied for the position estimation of IPMSMs. A {\em de facto} standard approach is, first, to probe high frequency periodic signals into the motor stator terminal; then, extract the high-frequency components of stator currents to extract position information \cite{NAMbook}. Recently, we proposed in \cite{YIetal,YItcst} a quantitative analysis to the conventional signal injection-based methods from a frequency domain viewpoint, and also proposed a new design with enhanced accuracy.

Fig. \ref{fig:test_D-1} shows the comparison between the proposed design and the high-frequency (HF) extraction method in \cite{YIetal,YItcst}.\footnote{The motor and parameters selection coincides with those in \cite{YIetal}, with signal probed only into $\bv_\alpha$. The parameters in the new design are selected as $\alpha=50$ and $\gamma=50$.} Although the performances are similar, the proposed design is much easier to tune without any compensation. The new design is also applicable to the case of high speeds, but the one in \cite{YIetal} fails.
\begin{figure}[htp!]
  \centering
  \includegraphics[width=0.40\textwidth]{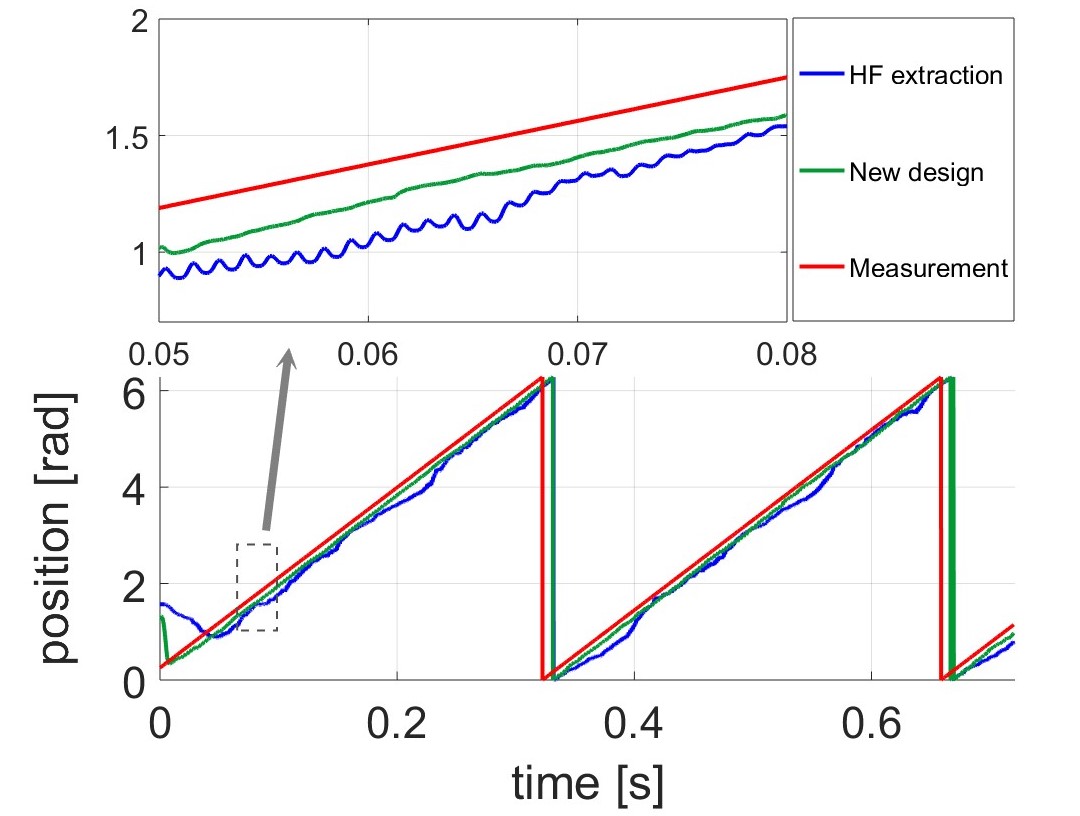}
  \caption{Comparison between the proposed observer with the one in \cite{YIetal}}\label{fig:test_D-1}
\end{figure}
%
\section{Concluding Remarks and Future Research}
\label{sec6}
%
The first GES flux/position observer for the practically important and theoretically challenging IPMSM was presented. The observer is designed following the ``linear regression plus gradient search" approach first proposed for PMSMs in \cite{ORTetalcst}, and succesfully pursued for IPMSMs in \cite{CHOetaltie,CHONAM}. GES of the observer is established, under a reasonable PE assumption, provided the filters used in the observer are ``not too fast" and the gain of the gradient search is ``sufficiently small". The latter conditions are imposed in the stability proof to be able to ``dominate" a disturbance term that appears in the GES part of the error system. Simulation results show that these conditions are not far from being necessary.

In \cite{CHONAM} the following observer is suggested
\begali{
\nonumber
\Dot{\hat \lambda} & = \bv - R\bi  + \gamma (\Phi + \Lambda) (y- \Phi^\top \hat\bx - \Hat{d})\\
\Hat{\bx} & = \hat\lambda - L_q i
\lab{newobs}
}
where
$$
\Lambda:= -\psi_m L_0 {\alpha p \over p + \alpha}\Big[
{1\over |\hat\bx|^3 }  \big( |\hat\bx|^2 I_2 - \hat\bx{\hat\bx}^\top \big) \bi
\Big].
$$
This observer is motivated by the fact that, for small $\alpha$,
$$
{\partial \over \partial \Hat{\bx}} (y- \Phi^\top \hat\bx - \Hat{d})^2=(\Phi + \Lambda) (y- \Phi^\top \hat\bx - \Hat{d}).
$$
Hence, \eqref{newobs} is a {\em bona fide} gradient search. Notice the absence of the matrix $\Lambda$ in the observer \eqref{observer} proposed in this paper. Unfortunately, the analysis of  \eqref{newobs} is a daunting task.

Current research is under way in the following directions.
\begite
\item To carry-out a simulation and experimental comparison of the observer proposed in this paper and \eqref{newobs}.
\item As the conditions of ``small" $\alpha$ and $\gamma$ of Proposition \ref{pro1} seem {\em necessary}, and this restrictions limit the transient performance of the observer, we are looking for some modifications to the observer to relax these conditions.
\item Although simulation and experimental results have shown that the observer can still be used---adding a probing signal---at standstill, further theoretical analysis is required to provide a solid foundation to this modification and, in particular, the transition from an active to a passive approach---see \cite{CHONAM} for a detailed discussion on this matter.
\item Another problem that has to be dealt with is output feedback control with the proposed observer, with needs to treat the closed-loop implementation both analytically and experimentally.
\endite

%
%

%
\begin{ack}                               
{The authors would like to thank three anonymous reviewers for their corrections and insightful remarks that helped to improve the quality of the paper.} This paper is supported by the Ministry of Science and Higher Education of the Russian Federation, project unique identifier RFMEFI57818X0271 ``Adaptive Sensorless Control for Synchronous Electric Drives in Intelligent Robotics and Transport Systems".
\end{ack}


\end{document}